\documentclass{article}
\usepackage{hyperref}
\usepackage{amsmath}
\usepackage{amsthm}
\usepackage{amsfonts}
\usepackage{booktabs}
\usepackage{hyperref}
\usepackage{algorithm}
\usepackage{algpseudocode}
\usepackage{graphicx}
\usepackage{subcaption}
\usepackage{multirow}
\usepackage{url}
\usepackage{fullpage}
\usepackage{svg}
\usepackage{tikz}
\usepackage{xcolor}
\definecolor{rAdjNorm}{RGB}{232,116,196}
\definecolor{APDA}{RGB}{28,116,180}
\definecolor{PC}{RGB}{212,36,44}
\definecolor{PopGo}{RGB}{156,104,192}
\definecolor{Reg}{RGB}{144,84,76}
\definecolor{LightGCN}{RGB}{52,164,44}
\definecolor{HetroFair}{RGB}{255,124,12}
\graphicspath{{./images/}}

\newtheorem{corollary}{Corollary}
\newtheorem{proposition}{Proposition}

\begin{document}
	\title{Heterophily-Aware Fair Recommendation using Graph Convolutional Networks}
	\author{Nemat Gholinejad, Mostafa Haghir Chehreghani\\
	Department of Computer Engineering\\
	 Amirkabir University of Technology (Tehran Polytechnic)\\
	 Tehran, Iran\\
	\texttt{\{n.gholinezhad,mostafa.chehreghani\}@aut.ac.ir}}
	\date{}
	
	\maketitle

		\begin{abstract}
	In recent years, graph neural networks (GNNs) have become a popular tool to improve the accuracy and performance of recommender systems. Modern recommender systems are not only designed to serve end users, but also to benefit other participants, such as items and item providers. These participants may have different or conflicting goals and interests, which raises the need for fairness and popularity bias considerations. GNN-based recommendation methods also face the challenges of unfairness and popularity bias, and their normalization and aggregation processes suffer from these challenges.
	In this paper, we propose a fair GNN-based recommender system, called HetroFair, to improve item-side fairness. HetroFair uses two separate components to generate fairness-aware embeddings:
	i) Fairness-aware attention, which incorporates the dot product in the normalization process of GNNs to decrease the effect of nodes' degrees.
	ii)	Heterophily feature weighting, to assign distinct weights to different features during the aggregation process.
	To evaluate the effectiveness of HetroFair, we conduct extensive experiments over six real-world datasets. Our experimental results reveal that HetroFair not only alleviates unfairness and popularity bias on the item side but also achieves superior accuracy on the user side. Our implementation is publicly available at \url{https://github.com/NematGH/HetroFair}.			
\end{abstract}

	\vspace{\baselineskip}
	
	\section{Introduction}
\label{sec:introduction}

Graph-based recommendation has become more popular in recent years due to the existence of a bipartite graph in the nature of recommender systems~\cite{gcmc, lightgcn, powerfulgraph, ngcf, dgcf, xie2021improving}. In this bipartite graph, nodes are divided into two sets: users and items, with edges representing ratings or interactions between them. Graph neural networks (GNNs) are a powerful method for deep learning in graph-structured data, as they can capture the complex and non-Euclidean relationships between the nodes and edges of a graph~\cite{zhou2020graph, gnnsurvey, halfdecade,DBLP:journals/tjs/ZohrabiSC24,DBLP:journals/corr/abs-2401-01384,10.1145/3700790}.

With the widespread use of graph neural networks in recommender systems, the competition to build systems with higher accuracy has entered a new phase~\cite{lightgcn, ecir2023item, ultragcn, yu2022low}. Early methods attempted to improve the accuracy of recommender systems to better satisfy end users. However, users are not the only stakeholders, and accuracy is not the sole objective. Recent research demonstrates that recommender systems are multi-stakeholder systems in which, in addition to end users, other participants (especially items and item providers) should also be considered during the recommendation process~\cite{survey_multistakeholder, multistakeholder, cp_ecir, cpfair, poi_unfairness}.

On the other hand, the power-law characteristic of the recommendation graphs shows that a large number of items possess a low degree, while only a few of them have a large degree. Figure \ref{fig:powerlaw} illustrates the degree distribution of two real-world recommendation datasets. This feature causes a phenomenon known in the literature as {\em popularity bias}, where popular items are recommended more frequently than other items. As a result, items with few interactions (long-tail items) are ignored, even if they are of high quality and relevance~\cite{abdollahpouri2019unfairness, naghiaei2022unfairness}.
Similar investigations have been done for GNN-based recommendation, revealing that long-tail items are often sacrificed during message normalization and aggregation of GNNs~\cite{adjnorm, apda}.

Various methods have been proposed in recent years for long-tail recommendation, which aim to recommend items that belong to the long tail of the distribution of ratings or popularity while simultaneously trying to achieve high accuracy on the user side. Regularization-based approaches attempt to mitigate popularity bias by changing the loss function that penalizes the entire network~\cite{popularitymetrics}. However, this often degrades overall performance, as it uniformly restricts the model without considering nuanced variations in item popularity. Another line of research focuses on multi-graph learning, where multiple graphs are optimized for different objectives in an integrated manner~\cite{hashtag,m2grl}. Additionally, causal learning has been leveraged to address popularity bias by constructing causal graphs to analyze the sources of bias and adjust the user-item relevance scores accordingly~\cite{wei2021model}.

Recent studies have shown that message passing in GNNs is a key factor in exacerbating bias against nodes with lower degrees~\cite{chen2024graph,bilateral}. Some works address this issue by modifying the normalization and aggregation processes in GNN-based recommender systems~\cite{adjnorm,apda}. However, these approaches overlook two critical aspects:
1)
The weights assigned to messages for normalizing the effect of popularity bias are static and non-trainable, remaining fixed across the entire propagation layers. This limits the model’s ability to adaptively adjust for different levels of item popularity.
2)
All features (attributes) of a node are treated equally, with a single fixed weight applied to all, despite the fact that different features capture distinct aspects of an item. For example, in movie-related item nodes, some features may represent the number of star actors, others the film producer, and others the genre, each contributing differently to the recommendation process.

In this paper, we propose the HetroFair model to improve item-side fairness in GNN-based recommendation systems. To do so, we first directly address the cause of popularity bias and reduce the effect of nodes' degrees on enlarging the values of their embeddings by involving dot products during the normalization process. Second, with respect to the heterophily property in graphs, we distinguish between different aspects of an item by considering a specific weight for each feature during the aggregation process. Our key contributions are summarized as follows:
\begin{itemize}
	\item
	We theoretically analyze the effect of message normalization in GNN-based recommendation systems on popularity bias.
	\item
	To generate fair embeddings for users and items, we design a fairness-aware attention mechanism by adding normalization terms derived from the dot product to embeddings.
	\item 
	We introduce feature-specific trainable weights to embeddings of users and items, integrated into the message-passing layers of GNNs, to compute heterophily-aware similarity between items and users.
	\item
	We conduct extensive experiments over six well-known datasets and show that our proposed model outperforms state-of-the-art methods in terms of both accuracy and fairness metrics. Furthermore, by performing ablation studies, we demonstrate that each component of our proposed method significantly contributes to the overall performance.
\end{itemize}

The rest of this paper is organized as follows.
In Section \ref{sec:relatedwork}, we provide a brief overview of related methods. In Section \ref{sec:preliminaries}, we introduce necessary preliminaries and definitions.
In Section \ref{sec:ourmethod}, we describe our proposed method in details.
In Section \ref{sec:experiments}, we report the results of our
extensive experiments conducted to show the high performance of our proposed method.
Finally, the paper is concluded in Section \ref{sec:conclusion}.

\begin{figure}[t]
	\begin{subfigure}[b]{0.48\linewidth}
		\includegraphics[width=\linewidth]{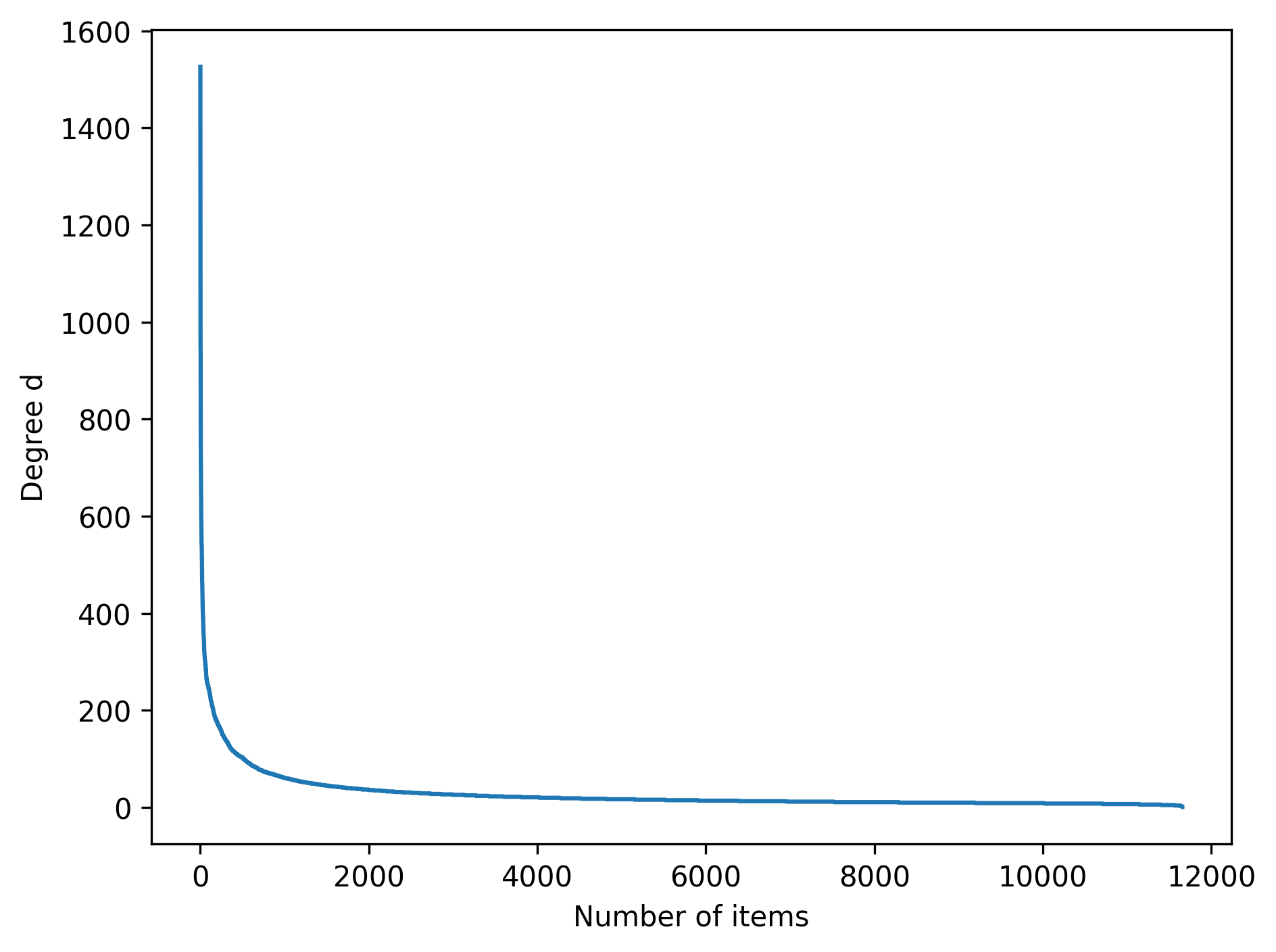}
		\caption{Epinions}
	\end{subfigure}
	\begin{subfigure}[b]{0.48\linewidth}
		\includegraphics[width=\linewidth]{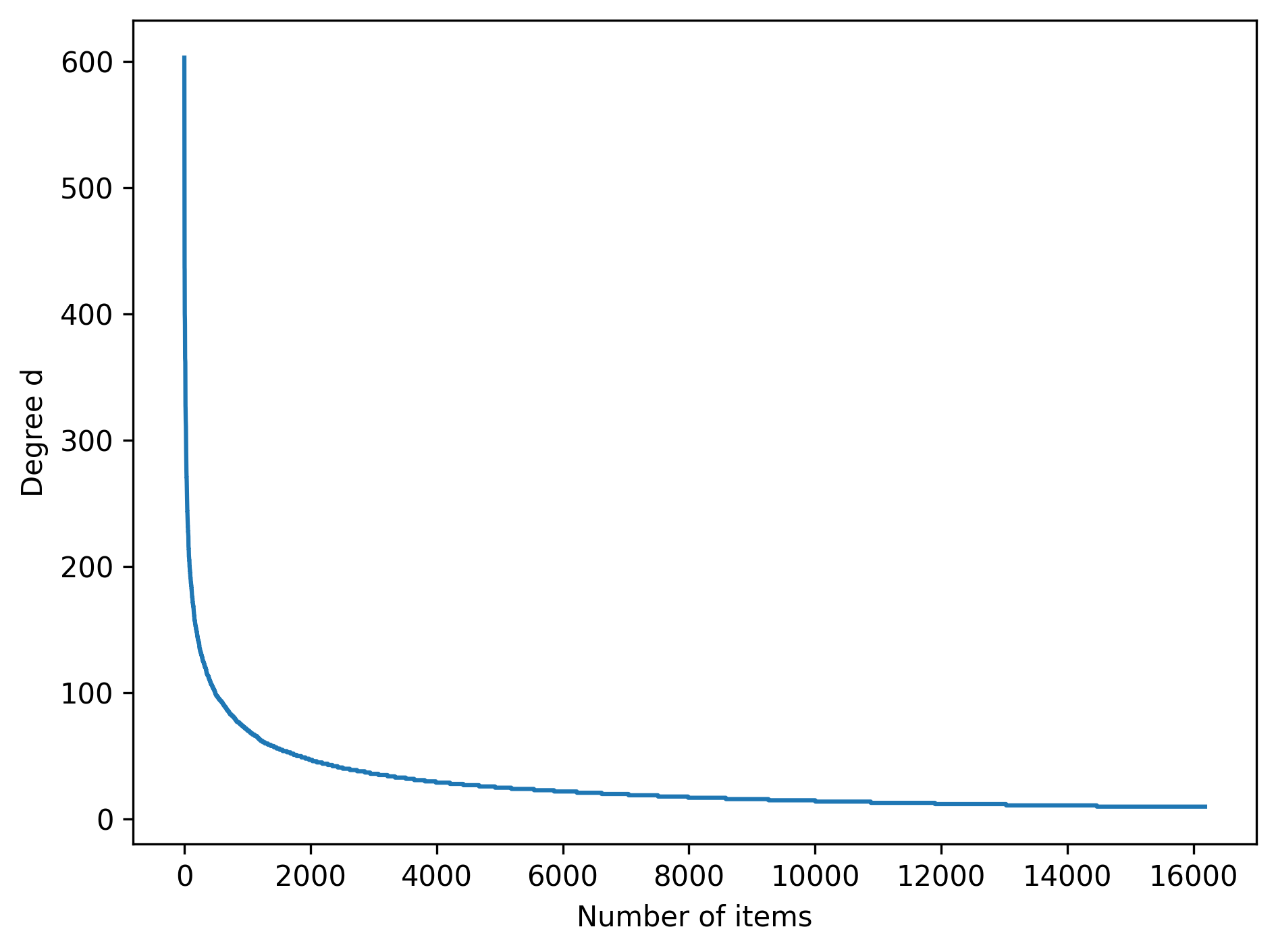}
		\caption{CDs}
	\end{subfigure}
	\caption{Degree distribution of items, in two recommendation datasets.}\label{fig:powerlaw}
\end{figure}

\section{Related work}
\label{sec:relatedwork}

In this section, we review recent advances in two fields that are closely related to our work:
graph-based recommendation and fairness-aware recommendation.

\subsection{Graph-based recommendation}

After the distinct capability of GNNs in node classification
within graph-structured data~\cite{gcn,graphsage,gat},
the GC-MC method~\cite{gcmc} extends this model to link prediction. It utilizes a graph auto-encoder to fill the user-item interaction matrix. PinSage~\cite{pinsage} introduces a recommender system for large-scale graphs,
leveraging the inductivity property of graph neural networks,
a concept proposed by GraphSage~\cite{graphsage}.
NGCF~\cite{ngcf} models higher-order connectivities by combining embeddings from different layers after embedding propagation on the user-item bipartite graph. LR-GCCF~\cite{lrgccf} removes non-linearity and demonstrates that this removal enhances recommendation performance. Taking it a step further, LightGCN~\cite{lightgcn} removes non-linearity and weight transformation, introducing a lightweight version of graph convolution for embedding propagation in user-item bipartite graphs. UltraGCN~\cite{ultragcn},
bypasses explicit message passing. Instead, it aims to compute the limit of the infinite graph convolution layer by leveraging the constraint loss.
IGCCF~\cite{ecir2023item} introduces an inductive convolution-based algorithm for recommender systems without training individual user embeddings, by leveraging an item-item graph.
Despite the high performance of these graph-based recommendation systems in improving accuracy, they often neglect to investigate unfairness
resulting from popularity bias,
and fail to use techniques that simultaneously consider items' side.

\subsection{Fairness-aware recommendation}

Research in the domain of fairness in graph-based recommendation is divided into two categories: diverse recommendation and popularity debiasing.

\paragraph{Diverse recommendation}
Diverse recommendation means that the items in individual recommendation lists should be dissimilar and cover a broader range of categories. BGCF~\cite{bgcf} utilizes  the node-copy mechanism and connects users to higher-order item neighbors that belong to different categories than their direct neighbors, yet remain similar.
DGCN~\cite{dgcn} mitigates the influence of dominant categories while elevating the significance of disadvantaged categories within neighboring nodes, achieved through re-balanced neighborhood sampling. Additionally, it leverages adversarial learning to achieve category-free representation for items.
Isufi et al.~\cite{isufiipm} suggest jointly learning from nearest and furthest neighbor graphs to achieve accuracy-diversity trade-off for node representations. DGRec~\cite{dgrec} applies a sub-modular function for the neighbor selection process to recommend a diverse subset of items.
It also manipulates the loss function, increasing the loss value for long-tail items and, conversely, decreasing it for short-head items.

\paragraph{Popularity debias}
The primary objective of popularity debias is to unearth qualified but less interacted items, long-tail ones, which are items that have lower training samples than short-head items. V2HT~\cite{hashtag} constructs a heterogeneous item-item graph with four different types of edges to connect long-tails to popular items. SGL~\cite{sgl} brings self-supervised learning to the recommendation setting and creates an auxiliary task for the long-tail recommendation. Zhao et al.~\cite{zhao2022popularity} disentangle popularity bias into two benign and harmful categories and argue that item quality can also affect popularity bias. They propose a time-aware disentangled framework to identify harmful bias.
r-AdjNorm~\cite{adjnorm} changes the power of the symmetric square root normalization term
in graph neural networks to control the normalization process during neighborhood aggregation, aiming for improved results, especially for low-degree items. APDA~\cite{apda} assigns lower weights to connected edges during the aggregation process and employs residual connections to achieve unbiased and fair representations for users and items in graph collaborative filtering. PopGo~\cite{zhang2024robust} proposes a debiasing strategy for collaborative filtering models by estimating interaction-wise popularity shortcuts and adjusting predictions accordingly. It trains a separate shortcut model to learn popularity-related representations and assigns a shortcut degree to each user-item pair.
Chen et al.~\cite{chen2024graph} theoretically analyze the effect of graph convolution on popularity bias amplification and show that stacked multiple graph convolution layers accelerate users' movement towards popular items in the representation space.
Our proposed method takes a different approach to improve 
items' side fairness and on the one hand,
considers different weights for features during	the aggregation process.
On the other hand, it 
directly involves dot products in the normalization process.

\section{Preliminaries}
\label{sec:preliminaries}

In this section, we introduce preliminaries widely used  in the paper,
including task formulation, graph neural networks, the Light Graph Convolution model and the heterophily computation method in user-item bipartite graphs.

\subsection{Problem setup}

In graph-based recommendation, we have a set of users \(U = \{u_1,u_2,\ldots,u_{|U|}\}\) and a set of items \(I=\{i_1,i_2,\ldots,i_{|I|}\}\) that constitute the nodes \(V\) of a bipartite graph \(G=\left(V,E\right)\),
where \(V = U \cup I \) and \(E\) represents implicit interactions (such as purchases, views and clicks) between users and items. The interactions  matrix \(R \in \mathbb{R}^{|U|\times|I|}\) is defined such that if \(e_{u,i} \in E\) then \(R_{u,i}=1;\) otherwise, it is 0. In graph \(G\), there exist no edges connecting users or items to each other. Thus, the graph adjacency matrix \(A\) is denoted as:
\[
A = \setlength{\arraycolsep}{3pt} % Set the margin between columns
\renewcommand{\arraystretch}{1.5} % Set the vertical spacing
\begin{bmatrix}
	0_{|U|\times|U|} & \hspace{1pt}\vrule\hspace{1pt} & R \\
	\hline
	R^{T} & \hspace{1pt}\vrule\hspace{1pt} & 0_{|I|\times|I|} \\
\end{bmatrix}.
\]
This matrix is used during message propagation by a graph neural network.

\subsection{Graph neural networks}

Graph neural networks (GNNs) are a class of learning algorithms that work well
on graph data structures.
The primary task of a GNN is to map each node in the graph to a vector in a low-dimensional vector space.
This vector is called the embedding or representation of the node.
GNNs rely on message passing, wherein each node in the graph sends a message to its direct neighbors.
Messages, in this context,
are intermediate vector representations of nodes.
Graph convolutional network (GCN)~\cite{gcn} is one of the first and most popular graph neural networks
that uses the following message passing rule:
\begin{equation}
	H^{\left(k+1\right)} = \sigma\left(\tilde{D}^{-\frac{1}{2}}\tilde{A}\tilde{D}^{-\frac{1}{2}}H^{\left(k\right)}M^{\left(k\right)}\right)\label{eq:gcn}.
\end{equation}
Here, \(D\)  is the diagonal degree matrix that has the same dimension as
the adjacency matrix \(A\) of the graph, \(D_{ii} = \sum_{j} A_{ij}\), \(\tilde{D}=D+I\), \(\tilde{A}=A+I\),  \(I \in R^{\left(|U|+|I|\right)\times\left(|U|+|I|\right)}\) is the identity matrix, \(\sigma\left(\cdot\right)\) is an activation function,
\(H^{\left(k\right)}\) is the matrix of node embeddings (representations) in the layer \(k\) and \(M^{\left(k\right)}\) denotes a layer-specific trainable weight matrix.
To compute the embedding of each node in the graph,
a graph neural network constructs a graph, which is a rooted tree and is usually called
node's {\em computational graph}~\cite{DBLP:conf/iclr/XuHLJ19}.

\subsection{Light Graph Convolutional}

The authors of LightGCN~\cite{lightgcn} demonstrate that non-linearity and feature transformation in the standard GCN are unnecessary for collaborative filtering graphs. They introduce a light version of GCN, as follows:
\begin{align}
	h_{u}^{\left(k+1\right)} =\sum_{i \in N(u)} \frac{1}{\sqrt{d_{u}}\sqrt{d_{i}}} h_{i}^{\left(k\right)},\label{eq:eq1}\\
	h_{i}^{\left(k+1\right)} =\sum_{i \in N(i)} \frac{1}{\sqrt{d_{u}}\sqrt{d_{i}}} h_{u}^{\left(k\right)}\label{eq:eq2}\cdot
\end{align}
Here, \(h_{u}^{\left(k\right)}\) and \(h_{i}^{\left(k\right)}\) denote the
embeddings of user \(u\) and item \(i\) at layer \(k\).
By \(N\left(x\right)\) we denote the set of nodes which have an edge to node \(x\),
and by \(d_x\) we denote the degree of node \(x\). The term \(\frac{1}{\sqrt{d_{u}}\sqrt{d_{i}}}\) is called the symmetric square root normalization term.
In the matrix form,
\(H^{\left(0\right)} = [h_{1}^{\left(0\right)}, h_{2}^{\left(0\right)}, \ldots, h_{|V|}^{\left(0\right)} ]\)
denotes the initial embeddings layer.
The embeddings in the next layers are formulated as follows:
\begin{equation}
	H^{\left(k+1\right)} = (D^{-\frac{1}{2}}AD^{-\frac{1}{2}})H^{\left(k\right)}.
	\label{eq:eq3}
\end{equation}

We define the size of an embedding $em$ as $\sqrt{em_1^2+ em_2^2+ \cdots + em_n^2}$, where $n$ is the dimension of $em$
and $em_i$ refers to its $i$-{th} element.
This is also known as the norm-2 of $em$. We say embedding $em$ is greater than embedding $em'$, denoted by $em > em'$, if the size of $em$ is greater than the size of $em'$.

\subsection{Homophily computation in bipartite user-item graphs}

To quantify homophily in the user-item bipartite graph, we use the following approach, which captures the alignment between users’ inferred preferences and the categories of the items they engage with.
Each item \(i \in I\) has an assigned label \(\ell(i) \in \mathrm{Label}\) denoting its category. Since users lack explicit labels, we infer a dominant preference label \(\ell(u)\) for each user \(u \in U\) by examining the labels of the items they have interacted with. Specifically, we define:
\[
\ell(u) = \arg\max_{\ell \in \mathrm{Label}}
\bigl|\{\,i \in N(u) \mid \ell(i) = \ell\,\}\bigr|,
\]
where \(N(u) = \{\,i \in I \mid (u, i) \in E\,\}\) is the set of items user \(u\) has interacted with, and \(\ell(i)\) denotes the label of item \(i\).
We then define a compatibility function \(\mathrm{compat}(u, i)\) that indicates whether an interaction \((u, i)\) is homophilic:
\[
\mathrm{compat}(u, i) =
\begin{cases}
	1, & \text{if }\ell(u) = \ell(i),\\
	0, & \text{otherwise}.
\end{cases}
\]

Using this compatibility function, we compute the homophily score \(H\) as the average compatibility across all valid edges (i.e., edges where both user and item labels are defined):
\[
H = \frac{1}{\lvert E' \rvert} \sum_{(u,i)\in E'} \mathrm{compat}(u,i),
\]
where \(E' \subseteq E\) contains only those interactions for which both \(\ell(u)\) and \(\ell(i)\) are defined.
The homophily score \(H\) offers insight into the structure of the interaction graph:
\begin{itemize}
	\item \(H \approx 1\): Strong homophily—users predominantly interact with items that match their inferred preferences.
	\item \(H \approx 0\): Strong heterophily—users interact with items outside their dominant label.
\end{itemize}
This metric enables us to distinguish between homophilic and heterophilic datasets and to evaluate model performance accordingly.

\section{Our proposed method}
\label{sec:ourmethod}

In this section, we first analyze the effect of symmetric square root normalization on popularity bias in GNN-based recommendation systems.
We then present our approach to overcome this bias. Subsequently, we outline the process of training our recommendation model.
Lastly, we provide the pseudocode of our proposed method and analyze its time complexity.
%	The main components of our designed model is presented in \ref{fig:proposed_method}

\subsection{Symmetric square root normalization effect}

The following theorem analyzes the over-smoothing problem in graph neural networks.
Our research investigates the implications of this foundational result in the context of fairness in graph-based recommendation systems.

\newtheorem{theorem}{Theorem}
\begin{theorem}
	\label{thm:oversmoothing}
	\cite{bilateral} For every undirected connected graph \(G=\left(V,E\right)\) where each node has a self loop, the following holds:
	\begin{equation}
		\small
		\lim_{{k \to \infty}} \left(\tilde{D}^{-1/2}\tilde{A}\tilde{D}^{-1/2}\right)_{i,j}^{k}=\frac{\sqrt{\left(d_{i}+1\right)\left(d_{j}+1\right)}}{2|E|+|V|}.
		\label{eq:eq5}
	\end{equation}
\end{theorem}

With respect to Light Graph Convolution in Equation \ref{eq:eq3},
after $k$ times message propagation,
the last layer node embeddings can be written as:
\begin{equation}
	H^{\left(k\right)} =\left(\tilde{D}^{-1/2}\tilde{A}\tilde{D}^{-1/2}\right)^{k}H^{\left(0\right)}.
	\label{eq:eq6}
\end{equation}

Combining Equation \ref{eq:eq6} with Equation \ref{eq:eq5} yields the following for user 
\(u\):
%	Equation \ref{eq:eq6} can be written as follows:
\begin{equation}
	h_u^{\left(k\right)} = \frac{\sqrt{d_{u}+1}}{2|E|+|V|} \left[ \left(\sqrt{d_{1}+1}\right)h_{1}^{\left(0\right)} + \left(\sqrt{d_{2}+1} \right) \right. \\
	\left. h_{2}^{\left(0\right)} +\ldots + \left(\sqrt{d_{|V|}+1}\right)h_{|V|}^{\left(0\right)}\right]\label{eq:eq7}.
\end{equation}

In Proposition \ref{prop:proposition1}, we exploit this equation to derive a relationship between the embedding sizes of users (and items) and their degrees. Then, in Proposition \ref{prop:proposition2}, we derive a relationship between the dot product similarity of the embeddings of items and users and their degrees. Later, in Corollary \ref{corollary1}, we combine these two to extract a relationship between the dot product similarities of the embeddings of items and users with their embedding sizes. In recommendation systems, the dot product is a commonly used operator to identify items that are relevant to users.

% this relation is used to investigate which items become more similar to a user $u$ (and hence are recommended to him), when the dot product is used as the similarity metric during recommendation.
\begin{proposition}
	\label{prop:proposition1}
	The embedding \(h_u^{(k)}\) of user \(u\) (respectively, the embedding \(h_i^{(k)}\) of item \(i\))
	is greater than the embedding \(h_{u'}^{(k)}\) of user \( u' \) (respectively, the embedding \(h_{i'}^{(k)}\) of item \(i'\)),
	if and only if the degree \(d_u\) of user \(u\) (respectively, the degree \(d_i\) of item \(i\))
	is greater than the degree \(d_{u'}\) of user \(u'\) (respectively, the degree \(d_i'\) of item \(i'\)).
\end{proposition}

\begin{proof}
	From Equation \ref{eq:eq7}, the embedding \(h_u^{(k)}\) for user \(u\) is given by:  
	\[
	h_u^{(k)} = \sqrt{d_u + 1} \times J,
	\]  
	where \(J = \frac{1}{2|E|+|V|} \sum_{v \in V} \sqrt{d_v + 1} \cdot h_v^{(0)}\) is a constant across all users. Similarly, the embedding for user \(u'\) is:  
	\[
	h_{u'}^{(k)} = \sqrt{d_{u'} + 1} \times J.
	\]
	
	Since \(J\) is constant and \(\sqrt{d_u + 1}\) is a monotonically increasing function of \(d_u\),  
	\[
	\sqrt{d_u + 1} > \sqrt{d_{u'} + 1} \quad \text{if and only if} \quad d_u > d_{u'}.
	\]
	
	Thus, \(h_u^{(k)} > h_{u'}^{(k)}\) if and only if \(d_u > d_{u'}\).  
	The same argument applies to item embeddings \(h_i^{(k)}\) and \(h_{i'}^{(k)}\), which are proportional to \(\sqrt{d_i + 1}\) and \(\sqrt{d_{i'} + 1}\), respectively.
	
\end{proof}

\begin{proposition}
	\label{prop:proposition2}
	The similarity score between user $u$ and item $i$,
	computed as the dot product \( h_u^{(k)} \cdot h_i^{(k)} \), is greater than the similarity score between user $u$ and item $i'$,
	represented by \( h_u^{(k)} \cdot h_{i'}^{(k)} \), if and only if the degree \( d_i \) of item \( i \) is greater than the degree \( d_{i'} \) of item \( i' \).  
\end{proposition}  

\begin{proof}  
	Inspired by Equations \ref{eq:eq7}, the similarity scores \( y_{ui} \) and \( y_{ui'} \) are given by:
	
	\[
	y_{ui} = h_u^{(k)} \cdot h_i^{(k)} = \frac{\sqrt{d_u + 1} \cdot \sqrt{d_i + 1}}{\left(2|E| + |V|\right)^2} \left( \sum_{v \in V} \sqrt{d_v + 1} \cdot h_v^{(0)} \right)^2,  
	\]
	
	\[
	y_{ui'} = h_u^{(k)} \cdot h_{i'}^{(k)} = \frac{\sqrt{d_u + 1} \cdot \sqrt{d_{i'} + 1}}{\left(2|E| + |V|\right)^2} \left( \sum_{v \in V} \sqrt{d_v + 1} \cdot h_v^{(0)} \right)^2.
	\]  
	
	Since the term  
	\[
	\left( \sum_{v \in V} \sqrt{d_v + 1} \cdot h_v^{(0)} \right)^2
	\]  
	is common in both expressions and \( \sqrt{d_u + 1} \) is identical, the comparison of \( y_{ui} \) and \( y_{ui'} \) depends only on the relationship between \( \sqrt{d_i + 1} \) and \( \sqrt{d_{i'} + 1} \). Therefore:  
	
	\[
	y_{ui} > y_{ui'} \quad \iff \quad \sqrt{d_i + 1} > \sqrt{d_{i'} + 1}.
	\]  
	
	Since the square root function is monotonically increasing, we have:  
	\[
	\sqrt{d_i + 1} > \sqrt{d_{i'} + 1} \quad \iff \quad d_i > d_{i'}.
	\]  
	
	Thus,  
	$y_{ui} > y_{ui'}$ if and only if $d_i > d_{i'}$.   
\end{proof}

\begin{corollary} 		
	\label{corollary1}
	The similarity score between user $u$ and item $i$, represented by the dot product \( h_u^{(k)} \cdot h_i^{(k)} \), is greater than the similarity score between user $u$ and item $i'$, $ h_u^{(k)} \cdot h_i'^{(k)}$, if and only if the embedding of item $i$ , $h_i^{(k)}$, is larger than the embedding of item $i'$, $h_i'^{(k)}$.
\end{corollary}

\begin{proof}
	From Proposition~\ref{prop:proposition1}, the embedding size \(h_i^{(k)}\) of an item \(i\) increases with its degree \(d_i\).
	Proposition~\ref{prop:proposition2} states that the similarity score \( h_u^{(k)} \cdot h_i^{(k)}\) is greater than \(h_u^{(k)} \cdot h_{i'}^{(k)}\) if and only if the degree \( d_i\) of item \(i\) is greater than the degree \(d_{i'}\) of item \(i'\).
	Since higher degrees correspond to larger embeddings, it follows that the similarity score \(h_u^{(k)} \cdot h_i^{(k)}\) is greater than \(h_u^{(k)} \cdot h_{i'}^{(k)}\) if and only if the embedding \(h_i^{(k)}\) is larger than \(h_{i'}^{(k)}\).
	
\end{proof}

\begin{figure}[!h]
	\begin{subfigure}[b]{0.33\linewidth}
		\includegraphics[width=\linewidth]{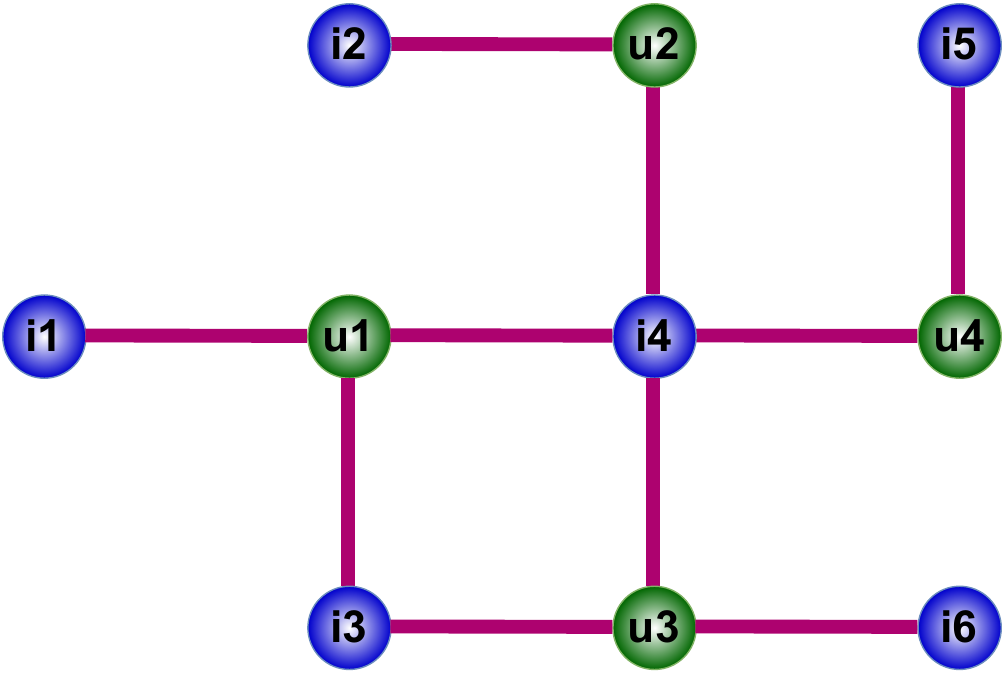}
		\caption{input graph}
	\end{subfigure}
	\begin{subfigure}[b]{0.33\linewidth}
		\includegraphics[width=\linewidth]{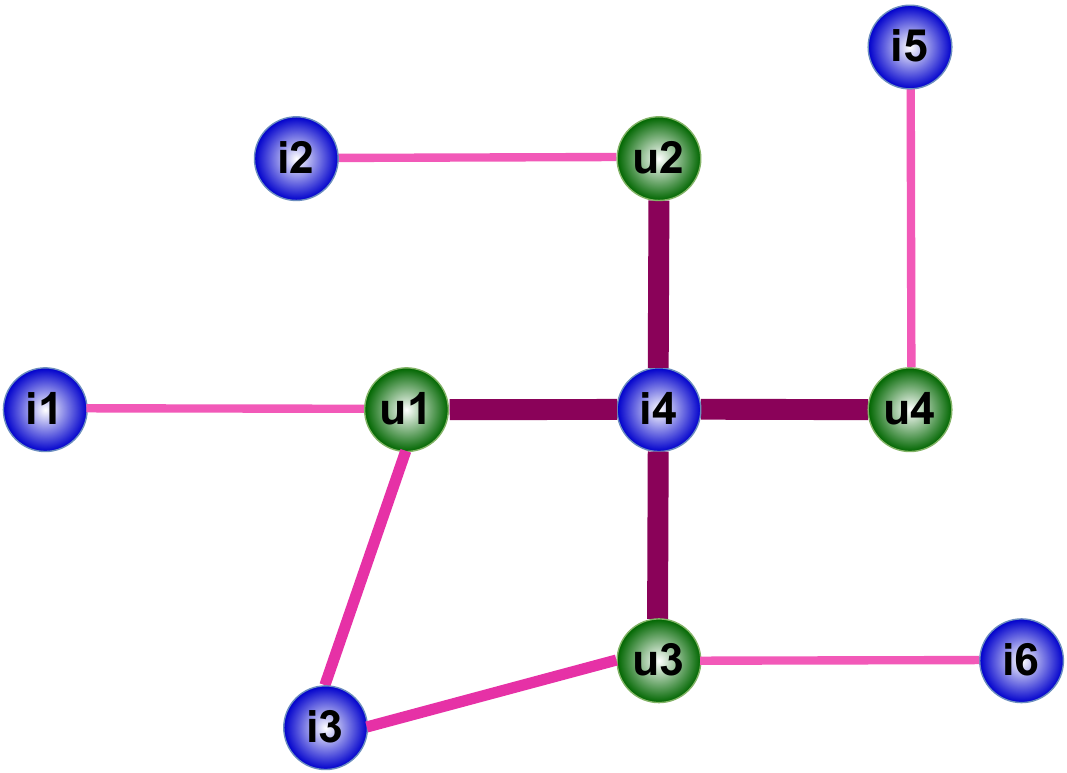}
		\caption{layer = 1}
	\end{subfigure}
	\begin{subfigure}[b]{0.33\linewidth}
		\includegraphics[width=\linewidth]{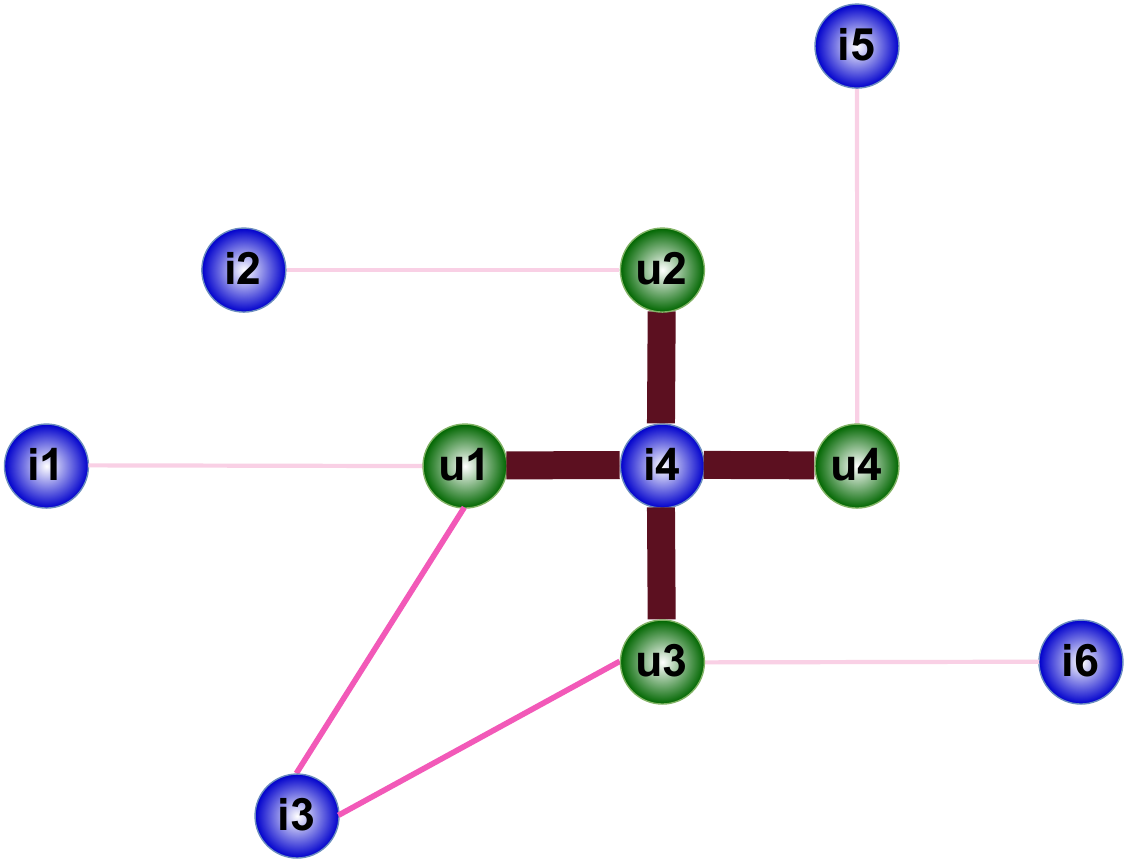}
		\caption{layer = 2}
	\end{subfigure}
	\caption{
		Illustration of the propagation process in a LightGCN model.}
	\label{fig:lightgcn_propagation}
\end{figure}

\begin{figure}[!h]
	\begin{subfigure}[b]{0.33\linewidth}
		\includegraphics[width=\linewidth]{images/lightlayer0}
		\caption{input graph}
	\end{subfigure}
	\begin{subfigure}[b]{0.33\linewidth}
		\includegraphics[width=\linewidth]{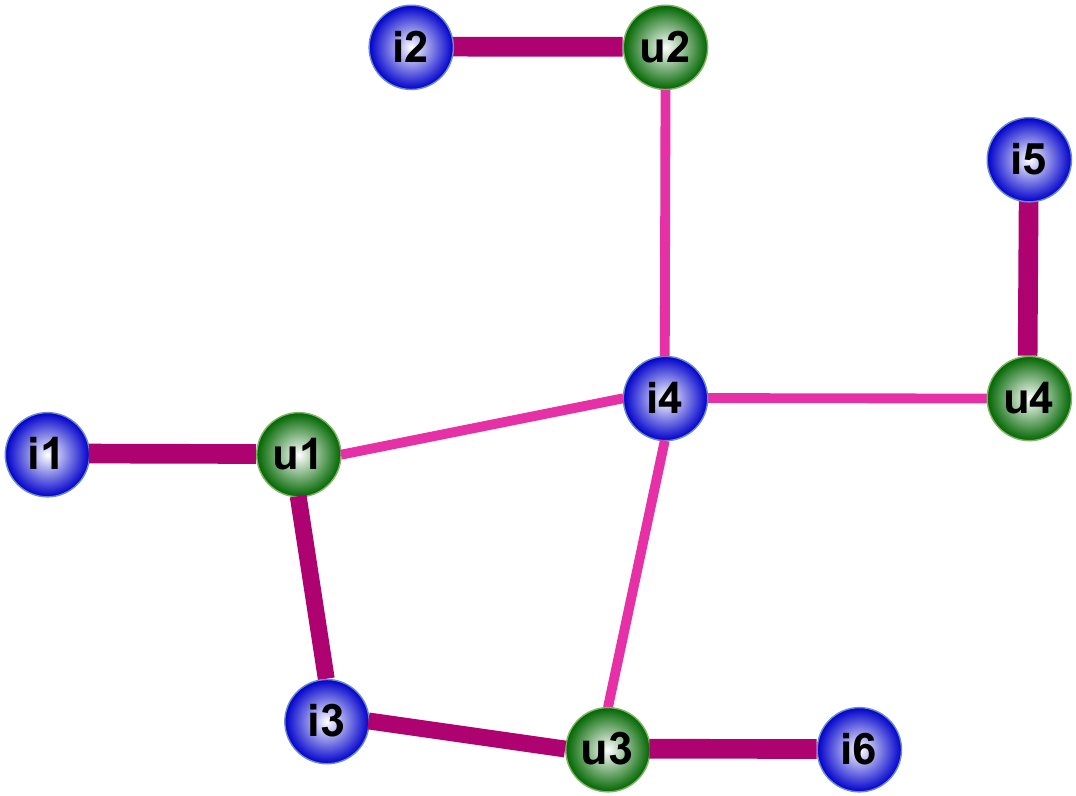}
		\caption{layer = 1}
	\end{subfigure}
	\begin{subfigure}[b]{0.33\linewidth}
		\includegraphics[width=\linewidth]{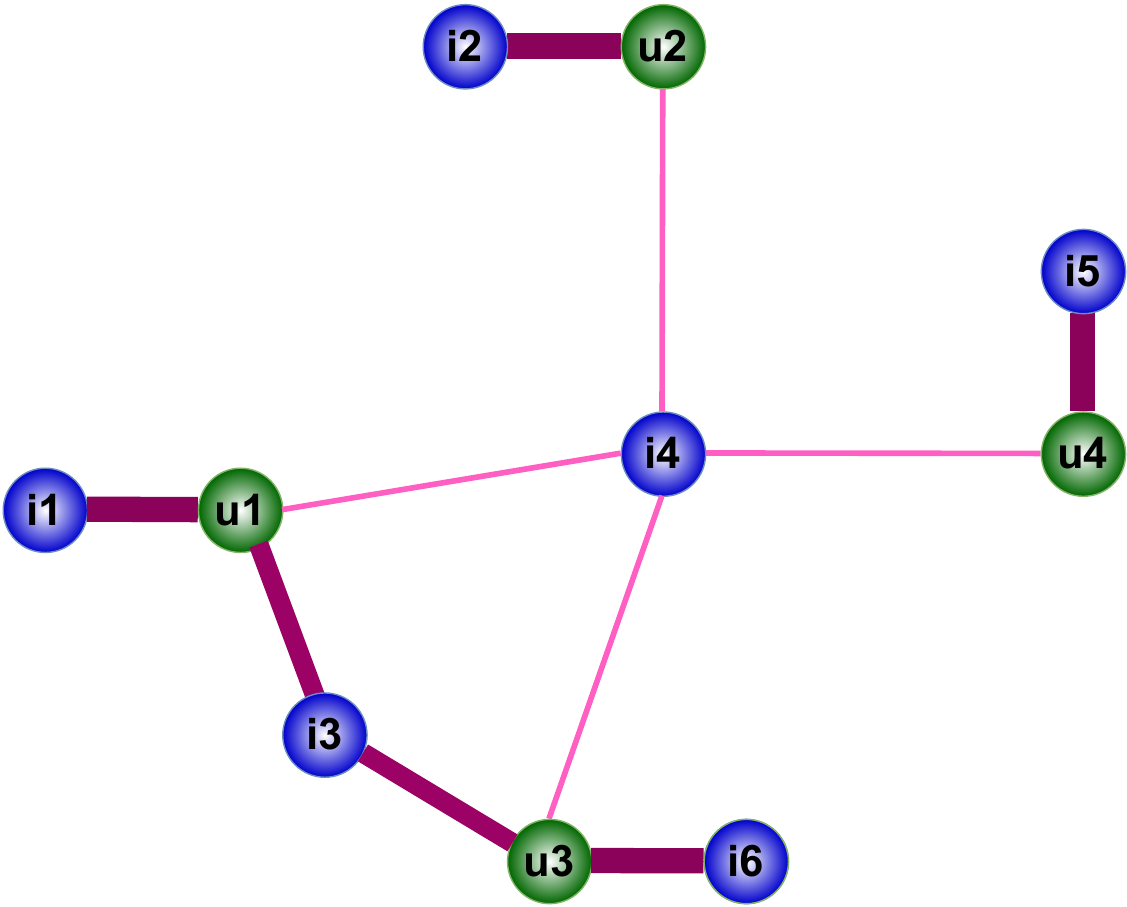}
		\caption{layer = 2}
	\end{subfigure}
	\caption{
		Illustration of the propagation process in our proposed method.}
	\label{fig:hetrofair_propagation}
\end{figure}

\subsection{Fairness-aware attention}
\label{sec:4.2}

In the previous section, we demonstrated the destructive effect of symmetric square root normalization on the representation of long-tail items:
according to Corollary \ref{corollary1}, nodes with larger embeddings have higher degrees and greater similarity scores, which harms fairness and exacerbates popularity bias by disproportionately favoring such nodes. To address this issue and prevent the embedding sizes from growing uncontrollably across multiple layers, it is crucial to control and reduce the embeddings' size.
To achieve this, we modify light graph convolution,
proposed by LightGCN~\cite{lightgcn}.
Specifically, since popularity bias arises from high degrees leading to large embeddings, which in turn result in higher similarity scores, this creates a recursive pattern. In our method, we incorporate the similarity scores from the previous layer as a normalization term in the current layer to balance the effects of degree and similarity. The modified version is as follows:

\begin{align}
	h_{u}^{\left(k\right)} =\sum_{i \in N\left(u\right)} \frac{1}{\sqrt{d_{u}}\sqrt{d_{i}} \times \delta \times \sigma(h_{u}^{\left(k-1\right)} \cdot h_{i}^{\left(k-1\right)})} h_{i}^{\left(k-1\right)},\label{eq:eq10}\\
	h_{i}^{\left(k\right)} =\sum_{u \in N\left(i\right)} \frac{1}{\sqrt{d_{u}}\sqrt{d_{i}} \times \delta \times \sigma(h_{i}^{\left(k-1\right)} \cdot h_{u}^{\left(k-1\right)})} h_{u}^{\left(k-1\right)},\label{eq:eq11}
\end{align}
where \(\sigma\left(\cdot\right)\) is the sigmoid function.
Parameter \(\delta\) plays a pivotal role in governing the extent to which
the item importance for user \(u\) is either amplified or diminished.
The sigmoid function transforms the dot product to a normalized range of \(\left[0,1\right]\). Therefore, the result of the sigmoid function for long-tail items will be close to zero, and for short-head items, it will be close to one. By incorporating the sigmoid of dot products along with the
above-mentioned normalization term, long-tail items achieve a smaller denominator than short-head items. This ensures that larger embeddings, which are more similar to each other, are divided by larger values and shrink more, while smaller embeddings are divided by smaller values and thus shrink less. Therefore messages from long-tail items (such as $i$) receive a higher importance coefficient than those from short-head items (such as $i'$). Consequently, this augmentation enhances the similarity between users and long-tail items,
compared to the previous case.
Conversely, messages from short-head items are propagated with lower coefficients.

This adjustment leads to changes in the graph structure obtained through successive propagation layers. Figures~\ref{fig:lightgcn_propagation} and \ref{fig:hetrofair_propagation} clearly illustrate these changes. In Figure~\ref{fig:lightgcn_propagation}, related to the LightGCN, we can see that during the application of different propagation layers, popular items establish stronger and closer connections with users, while the connections to less popular items weaken, symbolized by thinner edges. For example, these changes are shown by the thickening of the edges between users and the popular items $i_4$, and the thinning of the connections between user $u_1$ and item $i_1$. As we proceed through layers, this trend intensifies--high-degree items maintain dominance through robust connections, whereas low-degree items experience significantly reduced interaction influence, pushing them further from the users. In contrast, as shown in Figure~\ref{fig:hetrofair_propagation}, our proposed method strengthens connections between users and less popular items, making these edges thicker, while connections to popular items become relatively thinner, weakening their influence. Throughout the progression of layers, less popular items retain significant influence with strong connections to users, whereas popular items, though still linked, are pushed farther away from users and are less dominant in the graph structure. This not only prevents the over-enlargement of representations but also ensures that the effects of large representations are gradually transmitted through nodes' computational graphs.

\subsection{Heterophily feature weighting}
\label{sec:4.3}

The similarity function employed in the previous section is the dot product, a widely used measure for quantifying the alignment between two vectors. Despite its simplicity and computational efficiency, the dot product has a notable limitation: it is not trainable. This lack of trainability means it cannot adapt to the characteristics of the training data or optimize itself for specific tasks during the learning process. Consequently, it may fail to capture subtle relationships between users and items, which are critical for improving performance in recommendation tasks. To overcome this challenge, we propose a straightforward yet effective modification to the dot product, transforming it into a trainable similarity function. This enhancement enables the model to dynamically adjust similarity computations based on the training data.

In our approach, we introduce feature-specific weights to the user and item vectors, denoting them as $w_1$ for the first feature in users and items vector, $w_2$ for second feature and so forth. These weights are not predefined but are learned adaptively during the training process. This process is integrated into the message-passing layers, allowing the model to incorporate feature-specific adjustments at every stage of information propagation. We call this method the heterophily-aware similarity function because it operates at the feature level and assigns distinct, non-equal weights to capture heterophilous relationships. So,  instead of assigning a unique weight  \(\sigma(h_{u}^{\left(k\right)} \cdot h_{i}^{\left(k\right)})\) to the message between \(u\) and \(i\), we define the heterophily-aware feature weighting as follows:
\begin{equation}
	w_{ui}^{\left(k\right)} = \delta \times \sigma\left(\left(h_{u}^{\left(k-1\right)} \cdot h_{i}^{\left(k-1\right)}\right) W^{\left(k\right)}\right),\label{eq:eq12}
\end{equation}
where \(W^{\left(k\right)} \in \mathbb{R}^{1 \times d}\)
is a trainable linear transformation, \(d\) is the dimension of node embeddings and \(w_{ui}^{(k)} \in \mathbb{R}^{1 \times d}\).
Therefore, Equations \ref{eq:eq10} and \ref{eq:eq11} are revised as follows:
\begin{align}
	h_{u}^{\left(k\right)} =\sum_{i \in N\left(u\right)} \frac{1}{\sqrt{d_{u}}\sqrt{d_{i}}}
	h_{i}^{\left(k-1\right)} \oslash w_{ui}^{\left(k\right)} ,  \label{eq:eq13}\\
	h_{i}^{\left(k\right)} =\sum_{u \in N\left(i\right)} \frac{1}{\sqrt{d_{u}}\sqrt{d_{i}}}
	h_{u}^{\left(k-1\right)}  \oslash  w_{ui}^{\left(k\right)}, \label{eq:eq14}
\end{align}
where $\oslash$ is the Hadamard division operator.
For two vectors $a$ and $b$,
vector $c=a\oslash b$ is defined as follows:
$c[i]=a[i] / b[i]$.

Algorithm \ref{alg:alg1} outlines the fair embedding generation process. For every node in the graph and for each layer, we act as the follow.
First in Line~\ref{fairness-aware-attention}, we compute the fairness-aware attention as the dot products of the embeddings of users and items, for all items in the users' training profile.
Then, we transform these attention weights into a new vector space to consider distinct attention for each feature, as shown in Line~\ref{feature-weighting}. In Line~\ref{message-and-aggregate}, by taking the inverse of the weight in the message passing process, we enhance the likelihood of suggesting long-tail items to users while simultaneously diminishing the dominance of popular items. Finally, in Line~\ref{layer-combination}, we combine the embeddings from different layers to obtain
the final representations of items/users.

\begin{algorithm}[H]
	\caption{FairEmbeddingGeneration}\label{alg:alg1}
	\begin{algorithmic}[1]
		\State \textbf{Input:} Graph \(G\left(V,E\right)\); input features matrix $X \in \mathbb{R}^{|V|\times d}$; number of layers $K$; weight matrices $W^{\left(k\right)} \in \mathbb{R}^{1 \times d}, \forall k \in \{1,\ldots,K\}$; parameter $\delta$$\cdot$
		\State \textbf{Output:} Matrix representation $Z  \in \mathbb{R}^{|V|\times d}$ where each row in the matrix is the representation of node $v, \forall v \in V$.
		\State $h_v^{(0)} \gets x_v, \forall v \in V$
		\For{$v \in V$}
		\For{$k=1$ to $K$}
		\For{$i \in N(v)$}\label{neighbor-iterate}
		\State $s \gets h_v^{\left(k-1\right)} \cdot h_i^{\left(k-1\right)}$\label{fairness-aware-attention}
		\State $w_{vi}^{\left(k\right)} \gets \delta \times \sigma(\begin{bmatrix}s\end{bmatrix} \times W^{\left(k\right)}) $\label{feature-weighting}
		\Comment{By $\begin{bmatrix}s\end{bmatrix}$, we mean the $1 \times 1$ matrix, consisting of scalar $s$.}
		\EndFor\label{end-neighbor-iterate}
		\State $h_{v}^{\left(k\right)} \gets \frac{1}{\sqrt{d_v}} \cdot \sum_{i \in N(v)} \frac{1}{  \sqrt{d_{i}}} \cdot h_{i}^{\left(k-1\right)} \oslash w_{vi}^{(k)}$\label{message-and-aggregate}
		\EndFor
		\State $Z\left[v\right] \gets \frac{1}{K} \sum_{k=0}^{K} h_v^{\left(k\right)}$\label{layer-combination}
		\EndFor
		\State \textbf{return} $Z$
	\end{algorithmic}
\end{algorithm}

\subsection{Model training}

To optimize model parameters, we use the well known pairwise Bayesian Personalized Ranking (BPR) loss function~\cite{bpr}.
This loss function is designed to predict higher rank scores for items that users interact with, compared to unobserved items:
\begin{equation}
	L_{BPR} = \sum_{\left(u,i,j\right) \in T}-\ln\sigma\left(\tilde{y}_{ui} - \tilde{y}_{uj}\right) + \beta ||\Theta||^{2},
	\label{eq:eq15}
\end{equation}
where 
\(T = \left\{\left(u,i,j\right) \mid i \in O_{u}^{+} \wedge j \in I \setminus O_{u}^{+}\right\}\) indicates the training set and 	
\(O_u^+\) is user interacted items during training.
The hyperparameter \(\beta\) controls regularization to prevent overfitting, \(\Theta\) is the model's parameters that must be learned during the optimization process, \(\tilde{y}_{ui}\) and \(\tilde{y}_{uj}\) are user interest scores for items \(i\) and \(j\) respectively, obtained through the dot product of the embeddings of
the user and item.
\subsection{The algorithm}

In this section, we present the process of fair embedding generation,
along with the algorithmic steps that
our recommendation algorithm follows during a batch.
Algorithm~\ref{alg:alg2} describes one batch training process, in our model.
First, in Lines~\ref{start_init}-\ref{end_init}, we initialize model parameters using
the Xavier initialization method~\cite{xavier}.
The embedding matrix, updated during the training process, is also initialized.
In Line~\ref{fair_embedding}, we call the FairEmbeddingGeneration function,
presented in Algorithm~\ref{alg:alg1} to calculate users' and items' embeddings. Subsequently, for every user-item pair in the current batch $O_{batch}$,
we choose a negative sample uniformly at random with replacement,
as explained in~\cite{bpr}.
In Lines~\ref{score_predict}-\ref{backpropagate}, we compute the interest score of the user for positive and negative items, calculate the loss, and finally backpropagate the error to update model parameters.

\begin{algorithm}[H]
	\caption{One batch training process.\label{alg:alg2}}
	\begin{algorithmic}[1]
		\State \textbf{Input:} Graph \(G\left(V,E\right)\), the number of layers $K$.
		\State \textbf{Output:} Updated model parameters.
		\State $X \gets$ Xavier initialization\label{start_init}
		\Comment{Initialize input features}
		\For{$k=1$ to $K$}
		\State $W^{\left(k\right)} \gets$ Xavier Initialization
		\Comment{Initialize weight matrices}
		\EndFor\label{end_init}
		\State  $H \gets X$
		\Comment{Initialize embedding matrix}
		\State $L \gets 0$
		\Comment{Initialize BPR loss}
		\State $H \gets$ FairEmbeddingGeneration($G,H,K$)\label{fair_embedding}
		\For{$\left(u,i\right) \in O_{batch}$}
		\Comment{$O_{batch}$ is the set of $\left(u,i\right)$ pairs in the current batch}
		\State Sample a negative item $j$ uniformly at random from $I \setminus O_u^+$
		\State $\tilde{y}_{ui} \gets H\left[u\right] \cdot H\left[i\right]$\label{score_predict}
		\State $\tilde{y}_{uj} \gets H\left[u\right] \cdot H\left[j\right]$
		\State $L \gets L + \left(- \ln\sigma\left(\tilde{y}_{ui}-\tilde{y}_{uj}\right)+ \beta ||\Theta||^2\right)$
		\State \textbf{Backpropagate and update} model parameters for every node in $G$, using gradient descent on $L$\label{backpropagate}
		\EndFor
		\State \textbf{return} $\Theta$
	\end{algorithmic}
\end{algorithm}

\paragraph{Time complexity}
It is imperative for recommendation systems to deliver not only  satisfactory performance but also operate efficiency. Striking a balance between these two criteria is often viewed as a trade-off. In this context, we delve into the analysis of the time complexity of our proposed model.
If we neglect Lines~\ref{neighbor-iterate}-\ref{end-neighbor-iterate} of Algorithm \ref{alg:alg1}, our model
becomes the same as the LightGCN model, which has the time complexity of \(O\left(|E|Kd\right)\) in the graph convolution part,
where \(K\) and \(d\) can be seen as constants~\cite{graph-augment}.
Thanks to parameter sharing, the primary operation in Equation \ref{eq:eq12},
which appears in Lines~\ref{fairness-aware-attention}-\ref{feature-weighting},
is the multiplication of a $1\times 1$ matrix with a $1 \times d$ matrix,
which can be done in \(O\left(d\right)\) time.
Hence, the additional time complexity that our proposed method imposes on the basic LightGCN model is \(O\left(d\right)\).
According to this analysis, the time complexity of our proposed graph convolution process is \(O\left(|E|\left(Kd+d\right)\right) = O\left(|E|Kd\right)\),
which is the same as the LightGCN model.

\paragraph{Memory complexity}
In addition to runtime, memory consumption is an important factor for the scalability of GNN-based recommender systems. We analyze HetroFair's memory complexity from two perspectives: static parameter storage and dynamic training footprint.
Theoretically, the memory required for our model’s parameters is comparable to that of efficient baselines such as LightGCN. The additional feature-specific weight matrices \(W^{(k)}\) are shared across all nodes and scale only with the embedding dimension \(d\), contributing a negligible \(O(d)\) overhead to the parameter space. The primary static memory cost remains the storage of node embeddings and the graph structure, which is \(O(|V|d + |E|)\), consistent with standard GNNs.
However, it is important to distinguish this from the peak memory footprint during model training. Empirically, HetroFair’s training phase consumes more memory than LightGCN. This increase is not due to additional parameters but rather to the computational overhead of our fairness-aware attention mechanism. Computing dynamic, per-edge attention weights during message passing requires storing intermediate values for backpropagation, which is more memory-intensive than the sparse matrix multiplications used in LightGCN.
%	In addition to runtime efficiency, memory consumption is a critical factor for the scalability of graph-based recommender systems. Our model is designed to be memory-efficient and comparable in space complexity to existing GNN-based recommendation methods such as LightGCN.

We argue that the additional training cost represents a reasonable trade-off given the significant and consistent gains in both fairness and recommendation accuracy that our model delivers across six datasets. Moreover, this memory overhead is confined primarily to the training phase. During inference, when user and item embeddings can be pre-computed and served, the memory requirements are substantially lower and once again comparable to those of other GNN-based recommender systems.

\section{Experiments}
\label{sec:experiments}

In this section, we report the results of our extensive experiments, conducted over six real-world datasets, and demonstrate the high accuracy and performance of our proposed method.

\subsection{Experimental setup}

\subsubsection{Datasets}

We examine the algorithms over six well-known datasets:
Epinions\footnote{\url{https://snap.stanford.edu/data/soc-Epinions1.html}}~\cite{apda},
Amazon-Movies\footnote{\label{amazondata-footnote}\url{https://cseweb.ucsd.edu/~jmcauley/datasets/amazon/links.html}}~\cite{amazon-data}, Amazon-Electronics\textsuperscript{\ref{amazondata-footnote}}~\cite{amazon-data},
Amazon-Health\textsuperscript{\ref{amazondata-footnote}}~\cite{amazon-data},
Amazon-Beauty\textsuperscript{\ref{amazondata-footnote}}~\cite{amazon-data} and Amazon-CDs\textsuperscript{\ref{amazondata-footnote}}~\cite{amazon-data}.
For the amazon review datasets, we adopt a $10$-core setting,
filtering out users and items with less than $10$ interactions.
In the case of Epinions, we utilize the data released by APDA ~\cite{apda}.
Table \ref{tbl:dataset} depicts the statistics of the datasets.
The data highlight the diverse nature of our evaluation testbed, with significant variation in scale—from the smaller, denser Beauty and Health datasets to the larger, sparser Movies and Electronics datasets. This diversity is crucial for validation, as their density suggests that popularity bias and long-tail issues may manifest differently across them. Furthermore, the homophily rate spans a wide spectrum, from highly heterophilic (e.g., Movies at 0.14) to strongly homophilic (e.g., Health at 0.89), providing a robust environment for evaluating the effectiveness of our heterophily-aware components.
\begin{table}[!t]
	\centering
	\caption{Datasets statistics.\label{tbl:dataset}}
	\begin{tabular}{@{}l c c c c c @{}}
		\\
		\toprule
		Dataset & \#users & \#items & \#interactions & Density & Homophily rate\\
		\midrule
		Electronics & 20242 & 11589 & 347393 & 0.148\% & 0.15\\
		Epinions & 11496 & 11656 & 327942 & 0.245\% & N/A\footnotemark\\
		Beauty & 1340 & 733 & 28798 & 2.93\% & 0.79\\
		Health & 2184 & 1260 & 55076 & 2\% & 0.89\\
		Movies & 33326 & 21901 & 958986 & 0.131\% & 0.14\\
		CDs & 15592 & 16184 & 445412 & 0.176\% & 0.34\\
		\bottomrule
	\end{tabular}
\end{table}
\subsubsection{Baselines}
\footnotetext{The Epinions dataset lacks the category information necessary for computing the homophily rate.}
In order to evaluate the performance of HetroFair, we consider the following state-of-the-art recommendation models:

\begin{itemize}
	\item LightGCN~\cite{lightgcn}: This model simplifies graph convolution by excluding non-linearity and weight transformation. The final embedding is computed through a weighted sum of embeddings learned across various layers.
	Following the original paper, we use \(a_k = \frac{1}{K+1}\)
	as the coefficient of layer $k$, during layer combination.
	\item Reg~\cite{popularitymetrics}: This is a regularization-based method that considers the correlation between items' popularity and their predicted score by a model. Following the original paper, we tune the hyperparameter $\gamma$ carefully to achieve comparable results in accuracy and fairness.
	\item PC~\cite{popularitymetrics}: This is a pos-processing (re-ranking) approach that manipulates predicted scores to calculate new scores as a function of item popularity. We tune the parameter $\alpha$ and $\beta$ as described in the original paper.
	\item r-AdjNorm~\cite{adjnorm}: This model adjusts the strength of the normalization term to control the normalization process during neighborhood aggregation,
	aiming to improve results especially for low-degree items. Following the advice in the paper, we fine-tune the parameter \(r\) within the range \([0.5,1.5]\) with a step size of \(0.05\).
	\item APDA~\cite{apda}: This method assigns lower weights to connected edges during the aggregation process and employs residual connections to achieve unbiased and fair representations for users and items in graph collaborative filtering. Following the approach outlined in the original paper, we fine-tune the residual parameter \(\lambda\) within the range \(\left[0, 1.0\right]\).
	\item PopGo~\cite{zhang2024robust}: This method mitigates popularity bias by training a shortcut model to estimate interaction-wise shortcut degrees and adjusting CF model predictions accordingly. It quantifies shortcuts based on popularity representations of user-item pairs. Following the original paper, we fine-tune the hyperparameters 
	$\tau_1$ and $\tau_2$ within the recommended ranges to ensure optimal performance.
\end{itemize}
\subsubsection{Evaluation metrics}

\paragraph{Fainess metrics}
There are two distinct approaches to examine fairness on the items' side in recommender systems:

\begin{itemize}
	\item Statistical Parity (or Demographic Parity): this metric necessitates that different groups of items or similar individuals are treated similarly,
	without considering users' interests. This approach fails to consider the users' side and increases the risk of damaging users' satisfaction.
	\item Equal Opportunity: this metric mandates that groups of items benefiting from the outcome are afforded equal opportunities.
	It strives to treat all items fairly within users' profiles,
	regardless of their popularity.
\end{itemize}

The authors  of \cite{popularitymetrics} introduce two metrics based on equal opportunity:
\begin{align}
	PRU = -\frac{1}{|U|} \sum_{u \in U} SRC\left(d_{\tilde{O}_{u}^{+}}, rank_{u}\left(\tilde{O}_{u}^{+}\right)\right),\label{eq:eq16}\\
	PRI = -\sum_{i \in I} SRC\left(d_{i}, \frac{1}{U_{i}}\sum_{u \in U_{i}} rank_{u}\left(i\right)\right)\cdot\label{eq:eq17}
\end{align}
Here, \(SRC\left(\cdot,\cdot\right)\) calculates Spearman's rank correlation, \(\tilde{O}_{u}^{+}\) represents items in the test profile of user $u$, \(rank_{u}\left(x\right)\) signifies the ranking of item $x$ that the model predicts for user $u$, \(U_{i}\) is a set of users that item \(i\) already exists in their test profiles
and \(d_{x}\) is the degree of node \(x\).
While PRU investigates items' position bias for individual users' recommendation lists, PRI considers item-side fairness across all users. A large value of the later may result in long-tail items having limited opportunities to become popular \cite{popularitymetrics}. In general, the methods' performance in promoting fairness improves as the values of these two criteria get closer to $0$.

\paragraph{Accuracy metrics} As ranking holds paramount importance in recommender systems, we go beyond conventional machine learning metrics and incorporate rank-aware metrics to evaluate the performance of our model. Specifically, we employ MRR (mean reciprocal rank), MAP (mean average precision), and NDCG (normalized discounted cumulative gain) as accuracy metrics.
In the following, we provide a brief description of each metric.
NDCG denotes the normalized discounted cumulative gain and 
$NDCG@N$ for top $N$ recommended items is defined as follow:
\begin{equation}
	\small
	NDCG@N = \frac{DCG@N}{IDCG@N},\label{eq:eq18} 
\end{equation}
where \(DCG@N \) and \(IDCG@N\) are defined as follows:
{\small
	\begin{align}
		DCG@N = \sum_{i=1}^{N} \frac{2^{r\left(i\right)}-1}{log_{2}(i+1)},\label{eq:eq19}\\
		IDCG@N = \sum_{i=1}^{REL_N}\frac{2^{r\left(i\right)} - 1}{log_{2}(i+1)}.
		\label{eq:eq20}
	\end{align}
}
Here, \(REL_N\) is a sorted list of the $N$ most relevant items in ascending order, \(r(i)\) is the relevance of the \(i^{th}\) ranked item, which is either \(0\) or \(1\) respectively for irrelevant and relevant items.
In our experiments, we set $N=20$, and this value is applied consistently across all metrics.

MRR
is concerned with the position of the first relevant item:
\begin{equation}
	MRR = \frac{1}{|U|}\sum_{u \in U}\frac{1}{first_{u}},\label{eq:eq21}
\end{equation}
where \(first_{u}\) refers to the position of the first relevant item.
MAP is computed as the mean of average precision:
\begin{equation}
	MAP = \frac{1}{|U|}\sum_{u \in U} AP_u,\label{eq:eq22}
\end{equation}
where \(AP_{u}\) is the area under precision-recall curve.
$AP$ is calculated using the following formula:
\begin{equation}
	AP = \frac{1}{number\ of\ relevant\ documents}\sum_{i=1}^{N}P@i \cdot r(i),
	\label{eq:eq23}
\end{equation}
where \(P@i\) is precision at $i$.

\subsubsection{Hyper-parameters setting}

For all methods, we use the same early stopping mechanism, stopping training after \(15\) epochs if NGCF does not show  improvement in these \(15\) epochs.
The embeddings dimension is set to \(64\) for LightGCN and AdjNorm,
\(256\) for the APDA model, and \(128\) for our model.
The only exception is the 
Movies dataset where due to hardware limitations,
we set the embeddings dimension to 64.
The batch size for LightGCN and AdjNorm is set to \(1024\),
as proposed in their papers, while for APDA we set it to \(2048\).
The learning rate for all the baselines is set to \(0.001\).
For all the methods, we set both the regularization hyper-parameter \(\beta\)
and the weight decay parameter
to \(0.0001\).
The learning rate for our model is set to \(0.0005\).
We set the number of propagation layers to \(3\) in LightGCN and r-AdjNorm,
and to \(4\) in APDA and HetroFair.
For the PopGo model, the optimal number of layers varies between $2$ and $3$ across different datasets.
%	The embedding size is set based on paper $64$.
The embedding size is set to $64$.
For the two hyperparameter $\tau_1$ and $\tau_2$, various values are tested within the same range specified in the paper to achieve the best results.

\subsection{Performance comparison}
\subsubsection{Overall performance across metrics}
Table~\ref{tbl:performance_comparison} shows the results of HetroFair in comparison to the other methods.
Over each dataset and for each metric, the best performance is marked in bold, and the second-best performance is underlined.
The findings indicate that our method outperforms the other methods,
not only in the fairness metrics but also in the accuracy metrics.
Only in the two small and dense datasets (Health and Beauty),
the PopGo model performs well on certain metrics. However, our method also performs well on these datasets while achieving a better balance between fairness and accuracy.
These results show that our method is better fitted with the data and
takes into account their distinct features.
In the case of baselines, it is observable that although the Reg method often achieves commendable fairness metrics compared to other baselines, its accuracy is significantly lower than that of other methods. This suggests that regularization-based methods may reduce the performance of recommendations for end-users. Regarding the PopGo model, although it improves PRI in some cases, it fails to simultaneously enhance PRU, indicating its limited ability to adjust item rankings at the top of the recommendation list.

An important observation from Table~\ref{tbl:performance_comparison} is the consistent trade-off between accuracy and fairness among the baseline models. For instance, Reg often improves fairness metrics over the standard LightGCN, but this typically comes at the cost of accuracy (lower NDCG). This behavior reflects the inherent difficulty of mitigating popularity bias without degrading recommendation relevance. In contrast, our proposed HetroFair model consistently achieves a more favorable balance across datasets: it improves fairness metrics such as PRU and PRI (bringing them closer to zero) while simultaneously enhancing accuracy metrics. This suggests that our model not only addresses item-side unfairness effectively, but also leverages the underlying graph structure to maintain or improve overall recommendation quality.

\begin{table}[htbp]
	\centering
	\caption{Performance comparison of the models, over the used datasets.\label{tbl:performance_comparison}}
	\scalebox{0.85}{
		\begin{tabular}{c c c c c c c c c |c }
			\toprule
			Dataset & Metric &  LightGCN & Reg & PC & r-AdjNorm & APDA & PopGo & Our model & w/o Hetro \\
			\midrule
			\multirow{5}{*}{Epinions} & NDCG & 0.0789 & 0.0746 & 0.0778 & 0.0798 & 0.0859 & \underline{0.0879} & \textbf{0.0895} & 0.0889 \\
			& MRR & 0.1361 & 0.1288 & 0.1336 & 0.1369 & 0.1471 &  \underline{0.1515} & \textbf{0.1525} & 0.1522 \\
			& MAP & 0.0329 & 0.0307 & 0.0324 & 0.0339 & 0.0366 & \underline{0.0376} & \textbf{0.0379} & 0.0377 \\
			& PRU & 0.5233 & 0.5114 & 0.5168 & \underline{0.5077} & 0.5188 & 0.5976 & \textbf{0.4706} & 0.5054 \\
			& PRI & 0.5129 & 0.5065 & 0.5071 & 0.5366 & \underline{0.4983} & 0.5475 & \textbf{0.4384} & 0.4573 \\
			\midrule
			\multirow{5}{*}{Electronics} & NDCG & 0.0463 & 0.0407 & 0.0422 & 0.0461 & \underline{0.0517} & 0.0495 & \textbf{0.0525} & 0.0488 \\
			& MRR & 0.0645 & 0.0552 & 0.0610 & 0.0640 & \underline{0.0708} & 0.0687 & \textbf{0.0733} & 0.0676 \\
			& MAP & 0.0221 & 0.0188 & 0.0203 & 0.0219 & \underline{0.0250} & 0.0236 & \textbf{0.0256} & 0.0232 \\
			& PRU & 0.5216 & \underline{0.4779} & 0.4919 & 0.5300 & 0.5548 & 0.5208 & \textbf{0.4600} & 0.5020 \\
			& PRI & 0.5417 & \underline{0.5077} & 0.5325 & 0.5661 & 0.5692 & 0.5086 & \textbf{0.4462} & 0.4805 \\
			\midrule
			\multirow{5}{*}{CDs} & NDCG & 0.1231 & 0.1230 & 0.1245 & 0.1178 & \underline{0.1379} & 0.1228 & \textbf{0.1449} & 0.1346 \\
			& MRR & 0.1715 & 0.1723 & 0.1760 & 0.1638 & \underline{0.1940} & 0.1752 & \textbf{0.2017} & 0.1882 \\
			& MAP & 0.0611 & 0.0611 & 0.0617 & 0.0563 & \underline{0.0710} & 0.0607 & \textbf{0.0747} & 0.0680 \\
			& PRU & 0.3304 & 0.3287 & 0.3270 & \underline{0.3262} & 0.3265 & 0.4487 & \textbf{0.2972} & 0.3266 \\
			& PRI & 0.2824 & 0.3069 & 0.2974 & 0.3193 & 0.2763 & \underline{0.2530} & \textbf{0.2264} & 0.2495 \\
			\midrule
			\multirow{5}{*}{Health} & NDCG & 0.1259 & 0.1202 & 0.1224 & 0.1271 & \underline{0.1314} & 0.1311 & \textbf{0.1334} & 0.1351 \\
			& MRR & 0.1986 & 0.1902 & 0.196 & 0.2028 & \underline{0.2073} & 0.2044 & \textbf{0.2112} & 0.2093 \\
			& MAP & 0.0595 & 0.0570 & 0.0574 & 0.0610 & 0.0637 & \underline{0.0644} & \textbf{0.0656} & 0.0669 \\
			& PRU & 0.5162 & 0.4980 & \underline{0.4893} & 0.5270 & 0.5205 & \textbf{0.4075} & 0.4922 & 0.4830 \\
			& PRI & 0.5303 & 0.5116 & 0.5285 & 0.5240 & 0.5555 & \textbf{0.3997} & \underline{0.5028} & 0.4804 \\
			\midrule
			\multirow{5}{*}{Beauty} & NDCG & 0.2113 & 0.2161 & 0.2116 & 0.2136 & 0.2274 & \textbf{0.2382} & \underline{0.2308} & 0.2302 \\
			& MRR & 0.2525 & 0.2603 & 0.2540 & 0.2520 & 0.2790 & \textbf{0.2870} & \underline{0.2824} & 0.2798 \\
			& MAP & 0.1200 & 0.1240 & 0.1203 & 0.1235 & 0.1354 & \textbf{0.1415} & \underline{0.1364} & 0.1389 \\
			& PRU & 0.4428 & 0.4407 & 0.04385 & 0.4188 & 0.4418 & \textbf{0.4091} & \underline{0.4139} & 0.3117 \\
			& PRI & 0.2164 & 0.2071 & 0.2151 & 0.2089 & 0.2266 & \underline{0.1900} & \textbf{0.1821} & 0.1270 \\
			\midrule
			\multirow{5}{*}{Movies} & NDCG & 0.0672 & 0.0673 & 0.0690 & 0.0676 & \underline{0.0715} & 0.0682 & \textbf{0.0777} & 0.0741 \\
			& MRR & 0.0966 & 0.0957 & 0.0971 & 0.0955 & \underline{0.1011} & 0.0979 & \textbf{0.1093} & 0.1048 \\
			& MAP & 0.0304 & 0.0298 & 0.0308 & 0.0299 & \underline{0.0331} & 0.0314 & \textbf{0.0365} & 0.0340 \\
			& PRU & 0.3098 & \underline{0.3051} & 0.3082 & 0.3245 & 0.3384 & 0.5034 & \textbf{0.3045} & 0.3251 \\
			& PRI & 0.3951 & \underline{0.3714} & 0.3829 & 0.4436 & 0.3862 & 0.2818 & \textbf{0.3355} & 0.3590 \\
			\bottomrule
		\end{tabular}
	}
\end{table}

\subsubsection{Analysis by dataset homophily}
A deeper analysis of our results, when viewed through the lens of the datasets' homophily rates, reveals a compelling narrative about HetroFair's specific strengths. Our proposed model demonstrates its most significant performance gains on datasets characterized by low to moderate homophily, where user preferences are more diverse. For instance, on highly heterophilic datasets such as Movies and Electronics, and moderately heterophilic ones like Epinions and CDs, HetroFair consistently and substantially outperforms all baseline methods across nearly every accuracy and fairness metric. 
This suggests that its fairness-aware attention and heterophily-centric design are particularly effective at capturing signals from diverse user-item interactions.

Conversely, on the two most homophilic datasets, Health and Beauty, the performance gap narrows as the interactions become more uniform. Even in these settings, however, HetroFair remains highly competitive. For example, on the Health dataset, our model achieves the highest accuracy scores across all three performance metrics. On the Beauty dataset, it secures the second-best performance for accuracy while delivering the top performance for the PRI fairness metric. While other models like PopGo also show strong results on certain metrics in these environments, our model’s ability to consistently deliver a leading or competitive performance highlights its robustness. 
This overall pattern validates our model's core design principle: HetroFair is explicitly built to be heterophily-aware, and its capacity to excel in those conditions while remaining a top contender elsewhere confirms its effective and well-rounded architectural design.

\subsection{Ablation study}

In this section, we explore the impact of different components of 
our proposed model on its performance.
More precisely, 
we compare HetroFair against the situation wherein
the heterophily feature weighting component is not used.
We refer to this case as w/o Hetro. In addition to that we test different initialization strategies for the feature-specific weights.
The results of our analysis are available in
Table~\ref{tbl:performance_comparison} and Figure~\ref{fig:init_barplots}.. 

\subsubsection{Heterophily feature weighting}
To analyze the effect of this module, we disable the feature-specific weighting mechanism while keeping all other components unchanged, and evaluate its contribution to overall performance across accuracy and fairness metrics:
\begin{itemize}
	\item
	When we discard the heterophily feature weighting component,
	over the datasets Epinions, Electronics, CDs and Movies,
	all the metrics deteriorate.
	This observation indicates that the heterophily feature weighting component
	effectively improves both fairness and accuracy metrics.
	\item
	Over the Health and Beauty datasets,
	w/o Hetro achieves better results than HetroFair.
	One possible explanation for this phenomenon is over-smoothing.
	In the case of w/o Hetro,
	our experiments are conducted using 2 GNN layers,
	which is the optimal value for it.
	However,
	we run HetroFair with 4 GNN layers, which is its default value. The diameter of the Health and Beauty datasets is 8, a very small value to set the number of GNN layers to 4.
	It is known that in order to avoid over-smoothing in graph neural networks,
	the number of GNN layers should be considerably less than the diameter of
	the graph dataset.
	When for these two datasets we decrease the number of GNN layers from 4 to
	2, the performance of HetroFair significantly improves and it outperforms w/o Hetro.
	Over the Health dataset, the PRU and PRI metrics drop from \(0.4922\) and \(0.5028\) to \(0.4622\) and \(0.4794\), respectively,
	outperforming PRU and PRI of w/o Hetro.
	The accuracy metrics also increase, for example,
	NDCG increases from \(0.1334\) to \(0.1380\).
	The same trends hold for the Beauty dataset.
\end{itemize}

\subsubsection{Feature weight initialization}
We conduct an ablation study to evaluate the robustness of the model with respect to the initialization of the trainable feature-specific weights $\mathbf{W}^{(k)}$. Specifically, we compare three initialization strategies: zero initialization, random normal initialization, and Xavier initialization (used in the main model). As illustrated in Figure~\ref{fig:init_barplots}, we assess model performance across three datasets (Epinions, Health, and Beauty) and three key metrics: NDCG, PRU, and PRI.
While all three schemes yield comparable results, Xavier initialization consistently provides a more favorable balance between accuracy and fairness. Zero initialization tends to yield weaker fairness outcomes, likely due to the absence of early asymmetry in gradient flow, which limits the model’s ability to differentiate between features during early training. Random normal initialization introduces stochastic variability, but occasionally leads to less stable fairness metrics. Xavier initialization, by preserving the variance of input and output signals, avoids saturation and ensures more effective gradient propagation.
These observations indicate that the choice of initialization, although not dramatically altering overall trends, can affect stability and fairness. Xavier initialization offers a stable and generalizable solution, making it a reliable choice for initializing feature weights in fairness-aware GNN models.

\begin{figure}[!h]
	\centering
	
	\begin{subfigure}[b]{0.32\linewidth}
		\includegraphics[width=\textwidth]{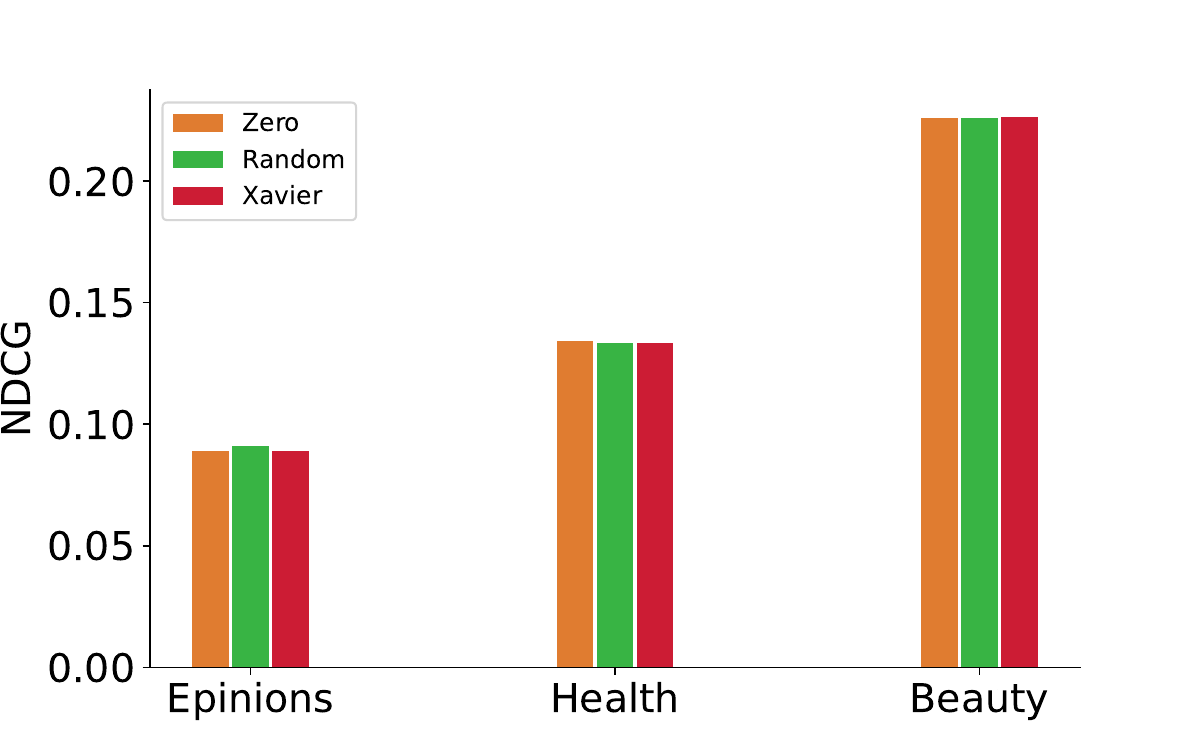}
		\caption{NDCG}
	\end{subfigure}
	\begin{subfigure}[b]{0.32\linewidth}
		\includegraphics[width=\textwidth]{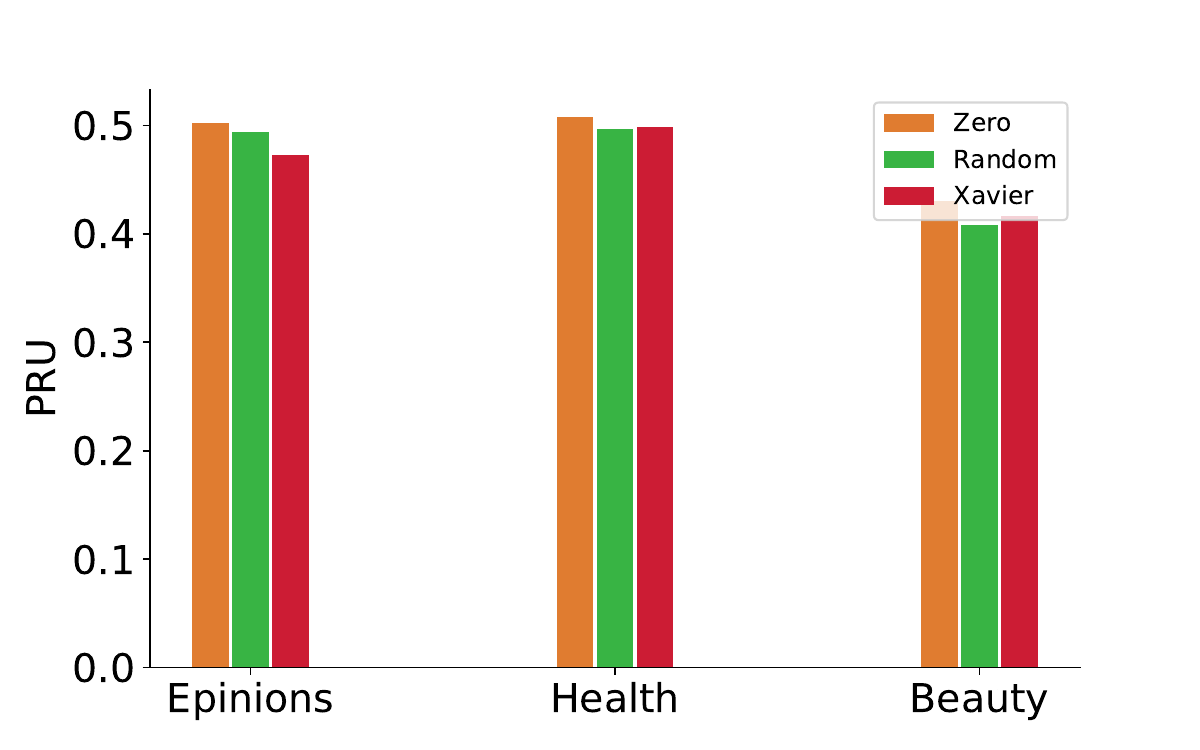}
		\caption{PRU}
	\end{subfigure}
	\begin{subfigure}[b]{0.32\linewidth}
		\includegraphics[width=\textwidth]{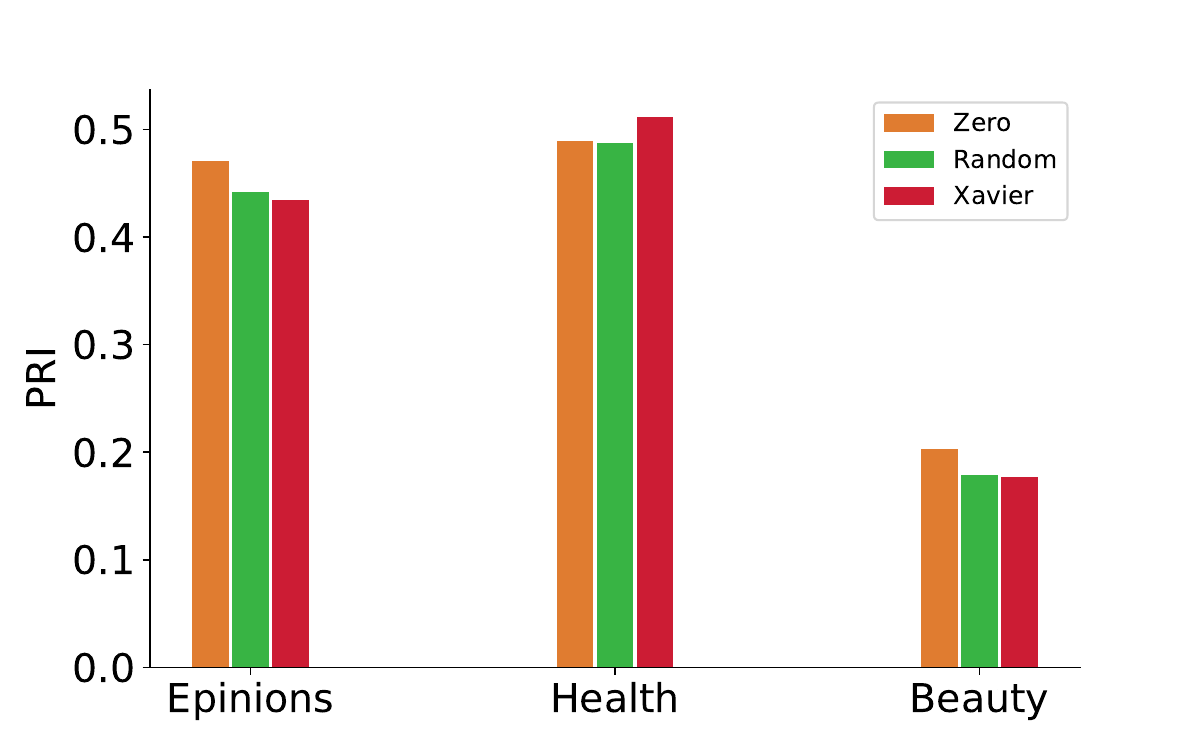}
		\caption{PRI}
	\end{subfigure}
	\caption{Performance of the model under different initialization strategies for feature-specific weights. .\label{fig:init_barplots}}
\end{figure}

\subsection{Hyper-parameters study}
\label{sec:hyperparameter_study}
In this section, we investigate the effect of different hyper-parameters on models' performance.
Specifically, we study the impact of three hyper-parameters:
the number of layers \(K\), the embeddings size \(d\)
and \(\delta\).
Due to the variety of the used datasets,
we study each of the parameters over  two datasets.

\paragraph{Number of layers}
To study the impact of the number of GNN layers, we select two datasets Beauty and Epinions, and investigate the potential effect of over-smoothing over them.
We evaluate the impact of the number of layers by ranging them from 1 to 6,
and use NDCG as the accuracy metric and PRU and PRI as the fairness metrics.
The results are presented in Figures \ref{fig:number_of_layer_epinions} and \ref{fig:number_of_layer_beauty}.
For the Epinions dataset, our model outperforms all the baselines across all the metrics.
In all the methods, NDCG increases until reaching the optimal point and then starts to decline due to the over-smoothing phenomenon. As the number of layers increases, PRU and 
PRI metrics decrease, indicating that increasing the number of layers up to a certain point can alleviate the unfairness caused by popularity bias.
The effect of over-smoothing becomes less noticeable due to the use of the weighted sum approach for the final node representation. On the other hand in the Beauty dataset, although we use the same weighted sum approach, the effect of the over-smoothing problem is evident from the initial stages of the models. The performance of the LightGCN and PC models drastically declines after 5 GNN layers.
The performance of HetroFair decreases too,
but with a smaller slope and it still outperforms all the
baselines over the Beauty dataset.
This shows its robust performance across different settings.

\begin{figure}[!h]
	\centering
	\vspace{0.5em}
	\begin{tikzpicture}[scale=1, every node/.style={scale=0.9}]
		\draw[APDA, thick] (0,0) -- (0.5,0) node[right]{\small APDA};
		\draw[HetroFair, thick] (2,0) -- (2.5,0) node[right]{\small Hetro-Fair};
		\draw[LightGCN, thick] (4.5,0) -- (5.0,0) node[right]{\small LightGCN};
		\draw[PC, thick] (7,0) -- (7.5,0) node[right]{\small PC};
		\draw[PopGo, thick] (8.5,0) -- (9.0,0) node[right]{\small PopGo};
		\draw[Reg, thick] (10.5,0) -- (11.0,0) node[right]{\small Reg};
		\draw[rAdjNorm, thick] (12,0) -- (12.5,0) node[right]{\small r-AdjNorm};
	\end{tikzpicture}
	
	\begin{subfigure}[b]{0.32\linewidth}
		\includegraphics[width=\textwidth]{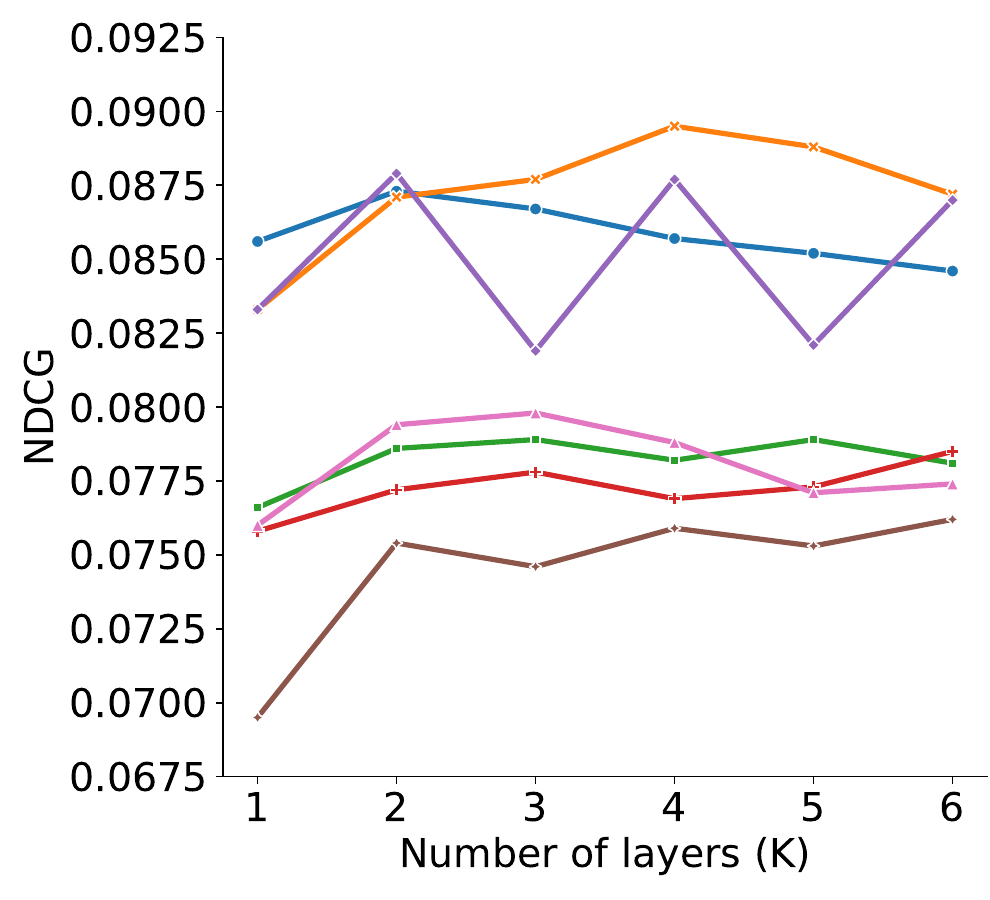}
		\caption{NDCG}
	\end{subfigure}
	\begin{subfigure}[b]{0.32\linewidth}
		\includegraphics[width=\textwidth]{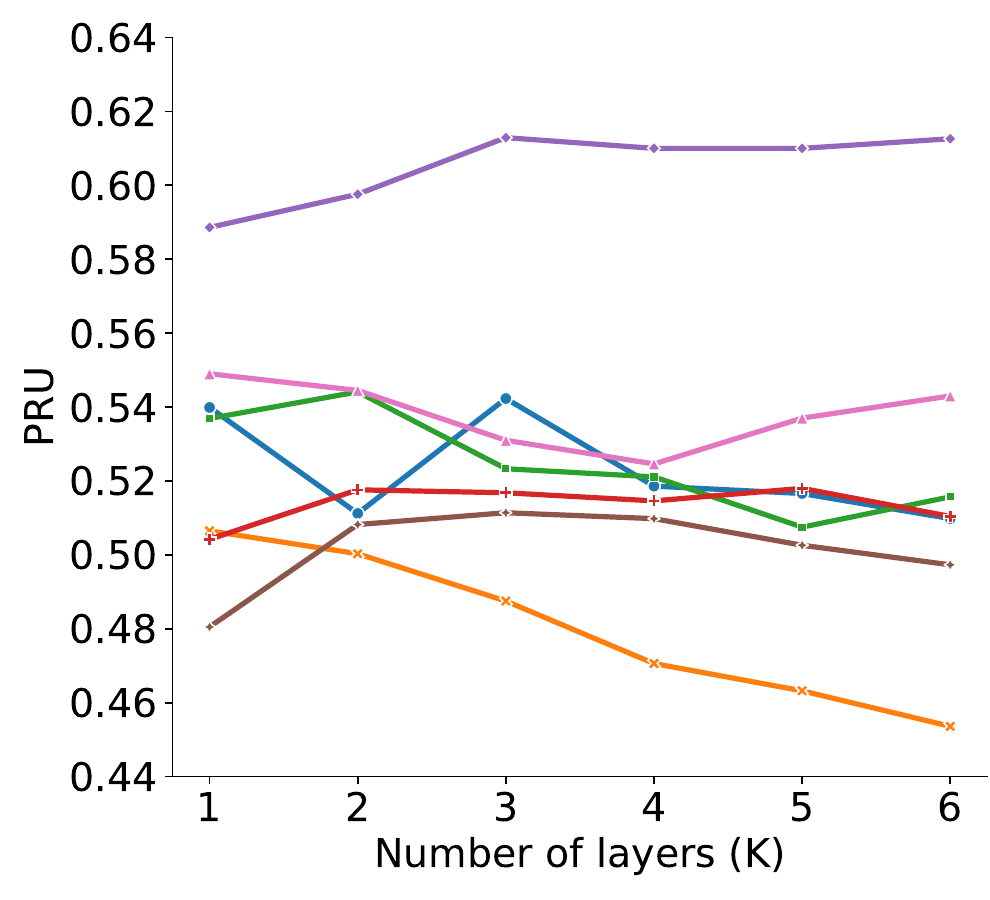}
		\caption{PRU}
	\end{subfigure}
	\begin{subfigure}[b]{0.32\linewidth}
		\includegraphics[width=\textwidth]{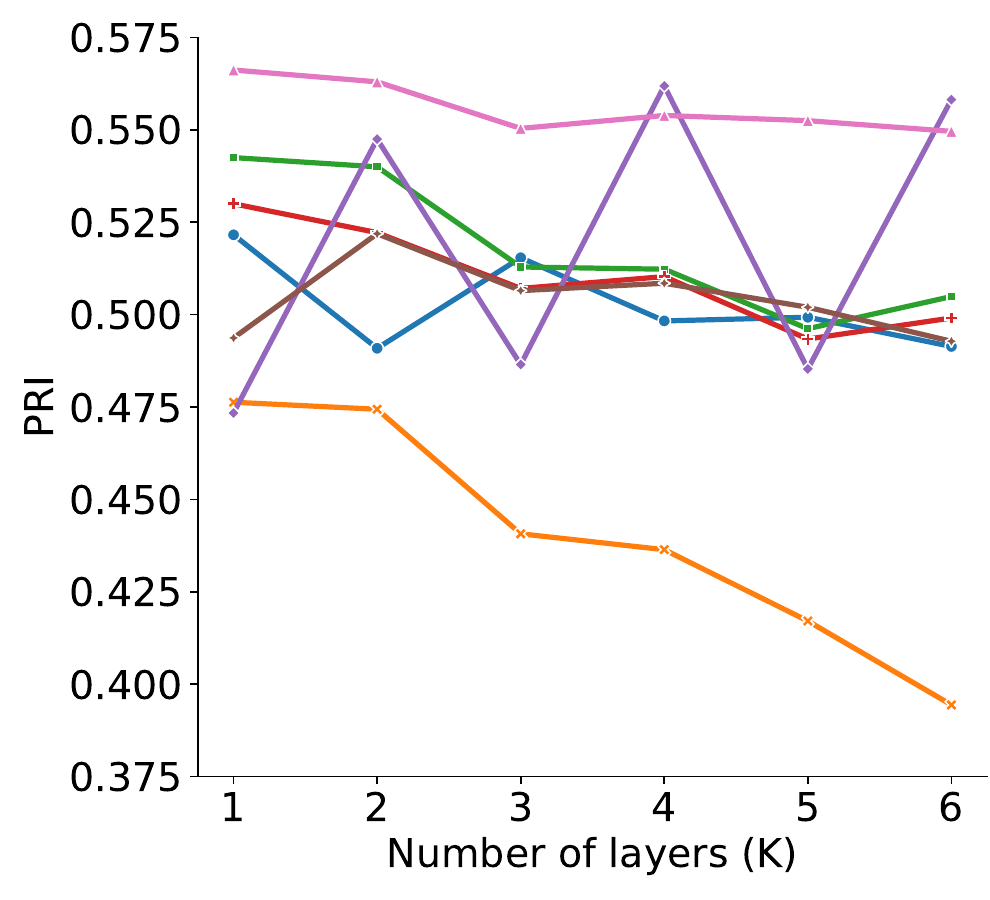}
		\caption{PRI}
	\end{subfigure}
	\caption{The performance of different models, over the Epinions dataset, for various numbers of message passing layers.\label{fig:number_of_layer_epinions}}
\end{figure}

\begin{figure}[!t]
	\centering
	\vspace{0.5em}
	\begin{tikzpicture}[scale=1, every node/.style={scale=0.9}]
		\draw[APDA, thick] (0,0) -- (0.5,0) node[right]{\small APDA};
		\draw[HetroFair, thick] (2,0) -- (2.5,0) node[right]{\small Hetro-Fair};
		\draw[LightGCN, thick] (4.5,0) -- (5.0,0) node[right]{\small LightGCN};
		\draw[PC, thick] (7,0) -- (7.5,0) node[right]{\small PC};
		\draw[PopGo, thick] (8.5,0) -- (9.0,0) node[right]{\small PopGo};
		\draw[Reg, thick] (10.5,0) -- (11.0,0) node[right]{\small Reg};
		\draw[rAdjNorm, thick] (12,0) -- (12.5,0) node[right]{\small r-AdjNorm};
	\end{tikzpicture}
	
	\begin{subfigure}[b]{0.32\linewidth}
		\includegraphics[width=\textwidth]{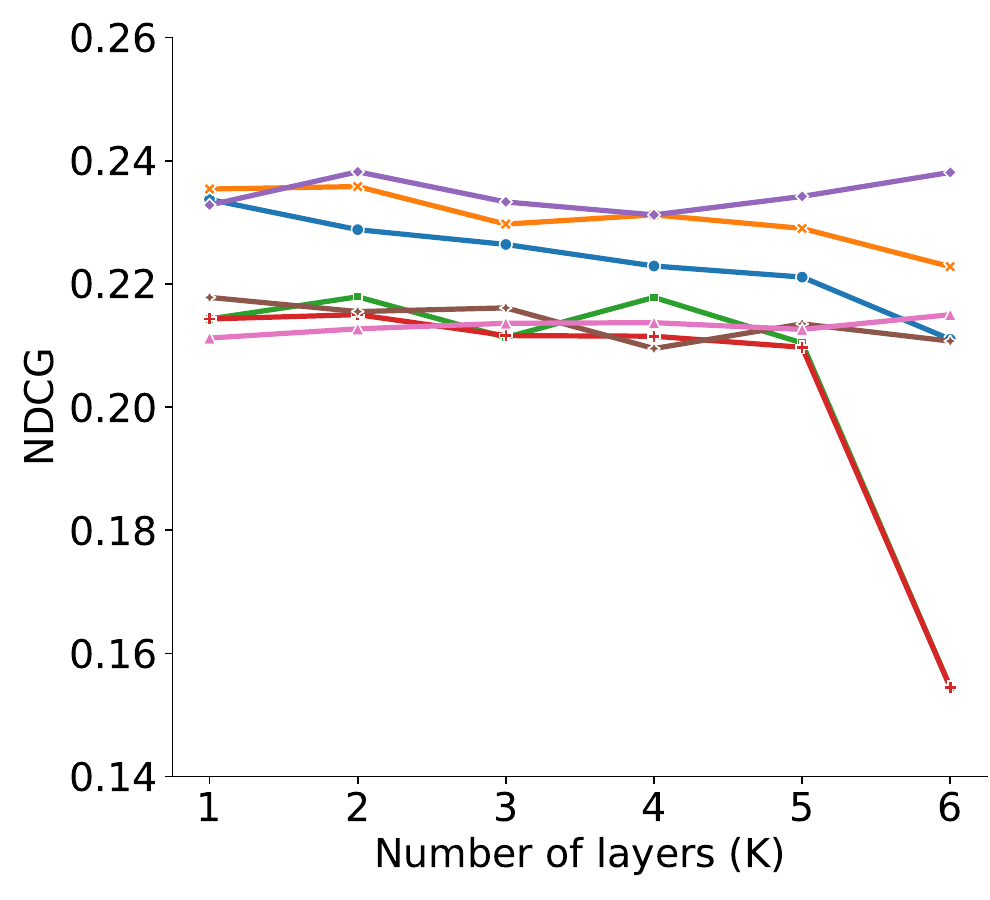}
		\caption{NDCG}
	\end{subfigure}
	\begin{subfigure}[b]{0.32\linewidth}
		\includegraphics[width=\textwidth]{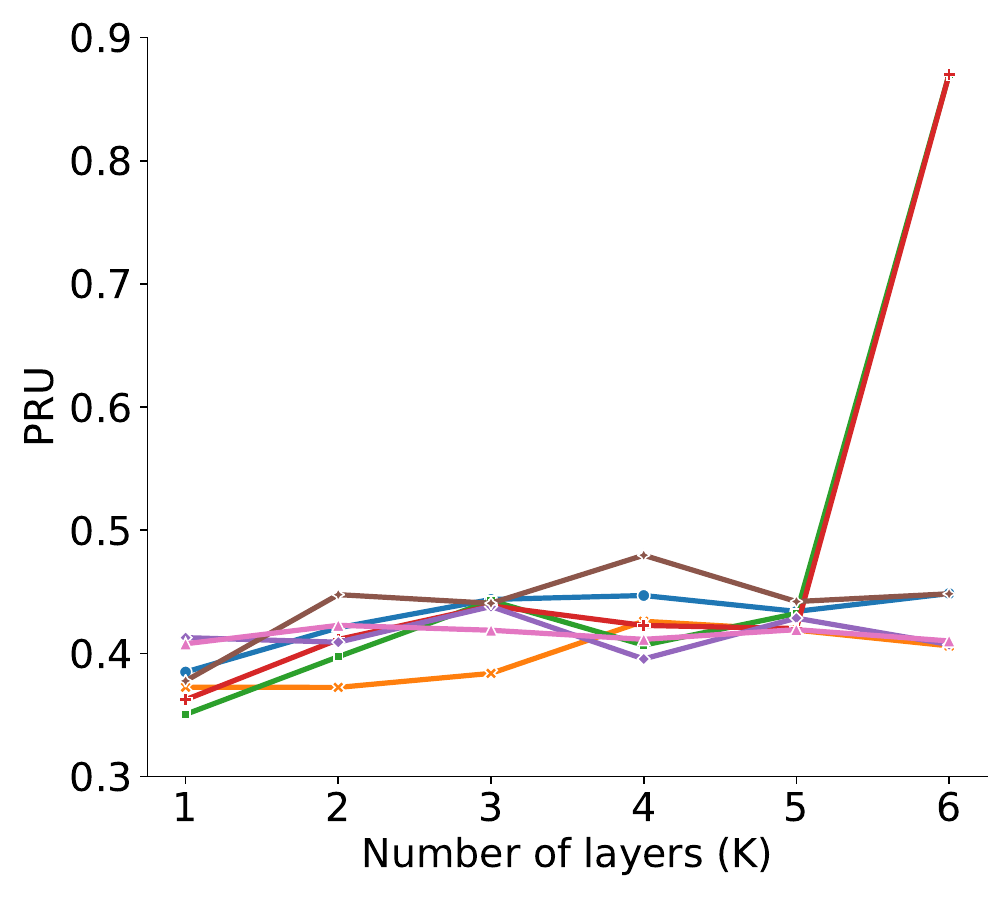}
		\caption{PRU}
	\end{subfigure}
	\begin{subfigure}[b]{0.32\linewidth}
		\includegraphics[width=\textwidth]{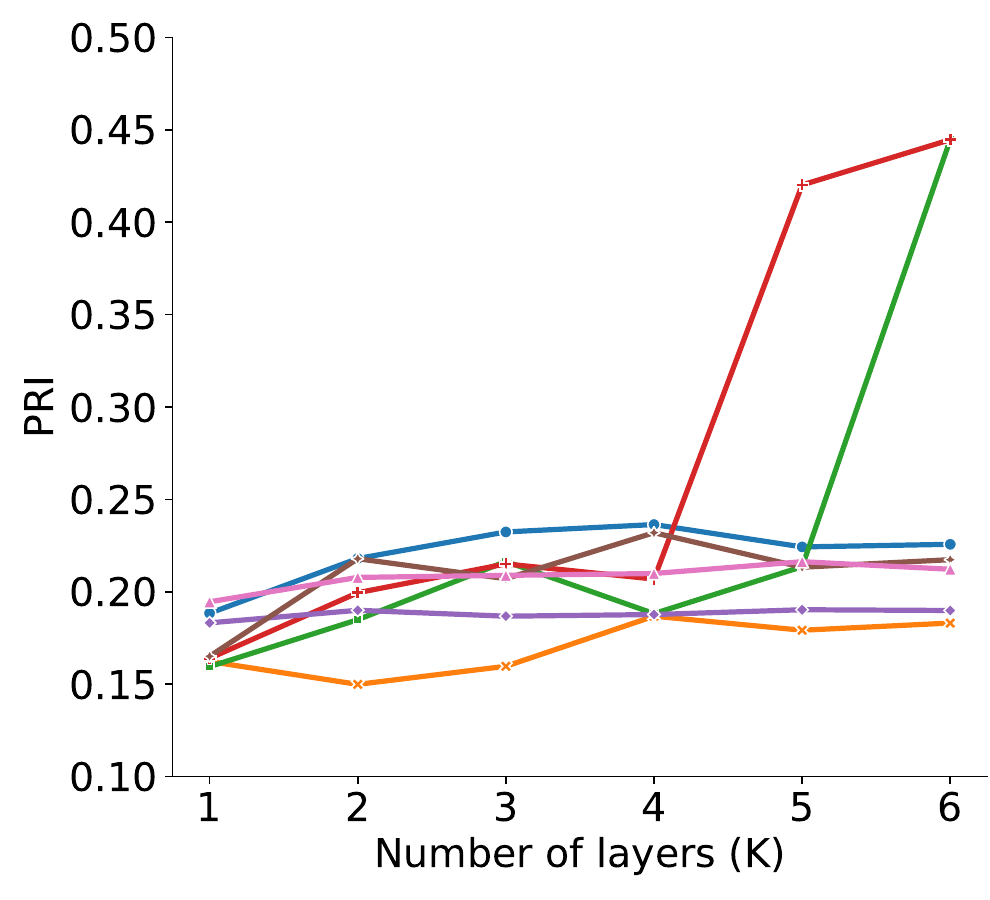}
		\caption{PRI}
	\end{subfigure}
	\caption{The performance of different models, over the Beauty dataset, for various numbers of message passing layers.\label{fig:number_of_layer_beauty}}
\end{figure}

\begin{figure}[H]
	\centering
	\vspace{0.5em}
	\begin{tikzpicture}[scale=1, every node/.style={scale=0.9}]
		\draw[APDA, thick] (0,0) -- (0.5,0) node[right]{\small APDA};
		\draw[HetroFair, thick] (2,0) -- (2.5,0) node[right]{\small Hetro-Fair};
		\draw[LightGCN, thick] (4.5,0) -- (5.0,0) node[right]{\small LightGCN};
		\draw[PC, thick] (7,0) -- (7.5,0) node[right]{\small PC};
		\draw[PopGo, thick] (8.5,0) -- (9.0,0) node[right]{\small PopGo};
		\draw[Reg, thick] (10.5,0) -- (11.0,0) node[right]{\small Reg};
		\draw[rAdjNorm, thick] (12,0) -- (12.5,0) node[right]{\small r-AdjNorm};
	\end{tikzpicture}
	
	\begin{subfigure}[b]{0.32\linewidth}
		\includegraphics[width=\textwidth]{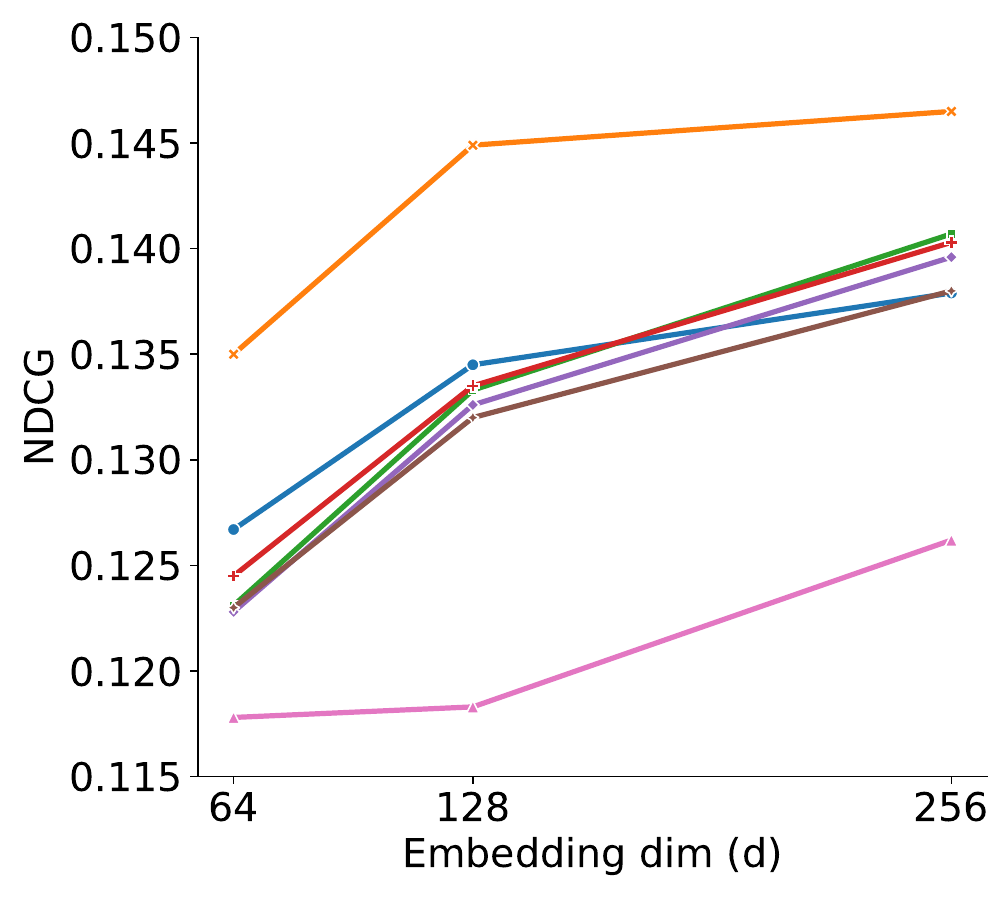}
		\caption{NDCG}
	\end{subfigure}
	\begin{subfigure}[b]{0.32\linewidth}
		\includegraphics[width=\textwidth]{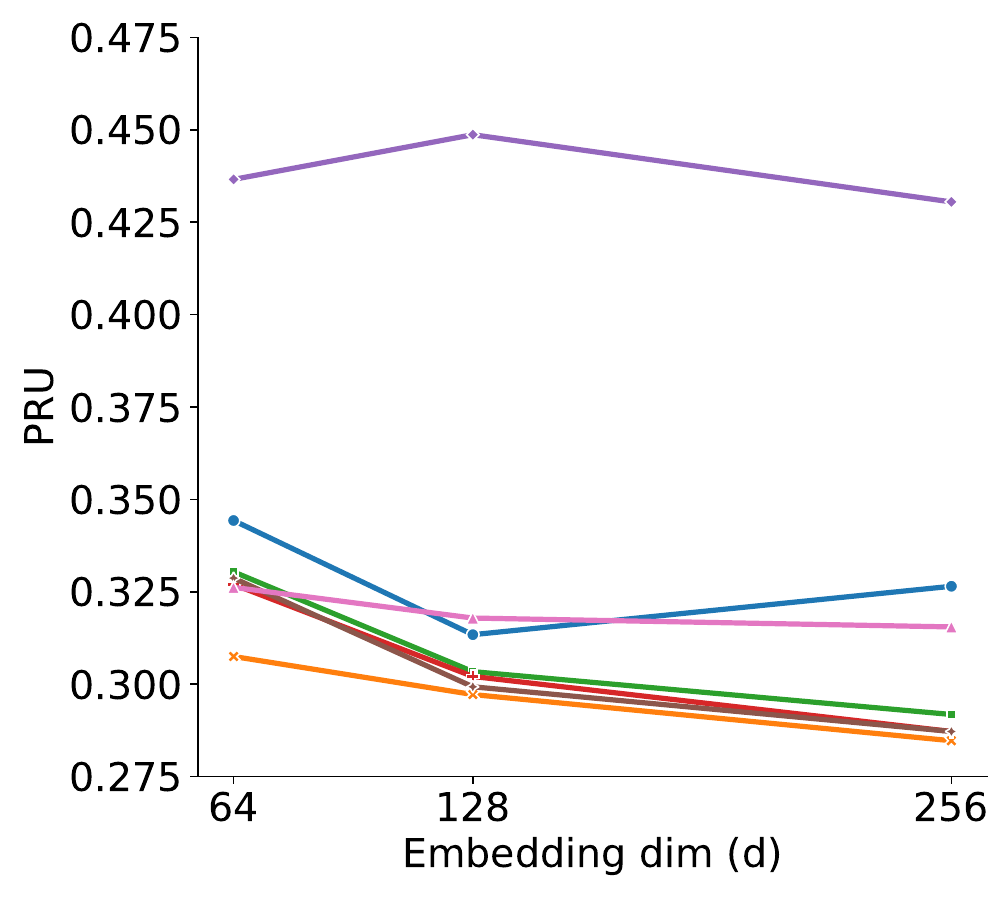}
		\caption{PRU}
	\end{subfigure}
	\begin{subfigure}[b]{0.32\linewidth}
		\includegraphics[width=\textwidth]{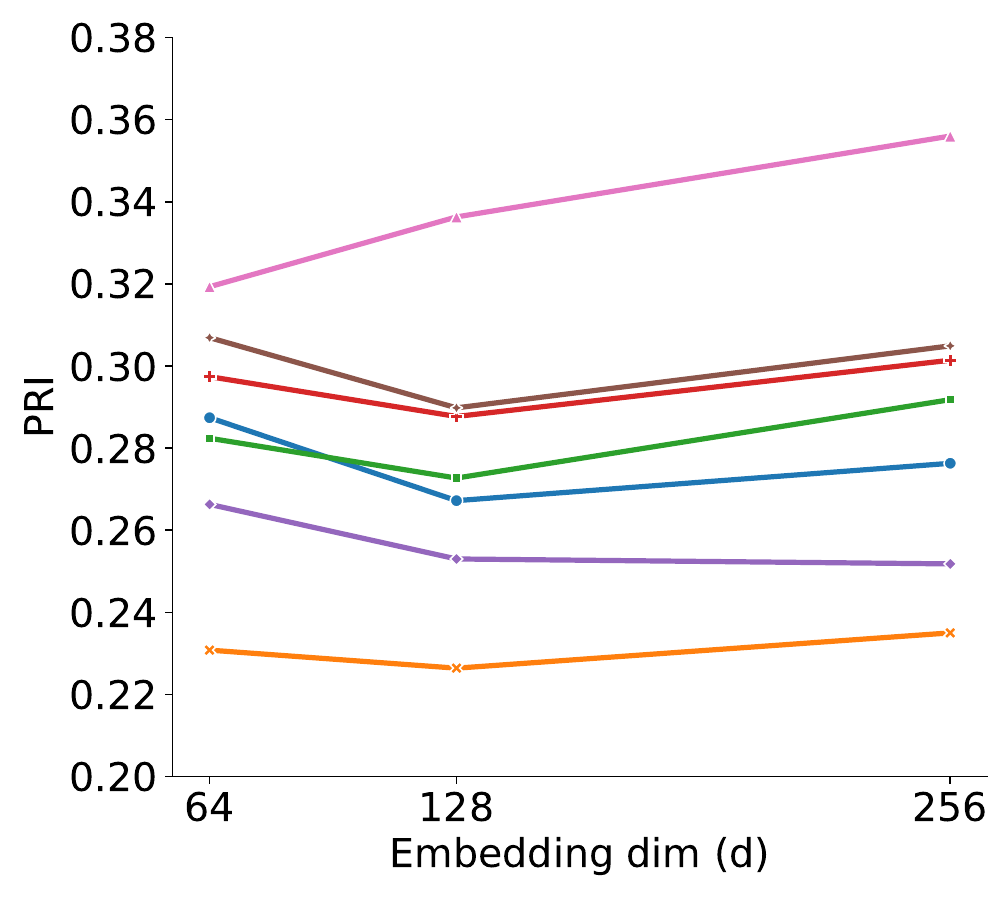}
		\caption{PRI}
	\end{subfigure}
	\caption{Comparing the models over the CDs dataset, for different values of embeddings dimension.\label{fig:embed_dim_cds}}
\end{figure}

\begin{figure}[H]
	\centering
	\vspace{0.5em}
	\begin{tikzpicture}[scale=1, every node/.style={scale=0.9}]
		\draw[APDA, thick] (0,0) -- (0.5,0) node[right]{\small APDA};
		\draw[HetroFair, thick] (2,0) -- (2.5,0) node[right]{\small Hetro-Fair};
		\draw[LightGCN, thick] (4.5,0) -- (5.0,0) node[right]{\small LightGCN};
		\draw[PC, thick] (7,0) -- (7.5,0) node[right]{\small PC};
		\draw[PopGo, thick] (8.5,0) -- (9.0,0) node[right]{\small PopGo};
		\draw[Reg, thick] (10.5,0) -- (11.0,0) node[right]{\small Reg};
		\draw[rAdjNorm, thick] (12,0) -- (12.5,0) node[right]{\small r-AdjNorm};
	\end{tikzpicture}
	
	\begin{subfigure}[b]{0.32\linewidth}
		\includegraphics[width=\textwidth]{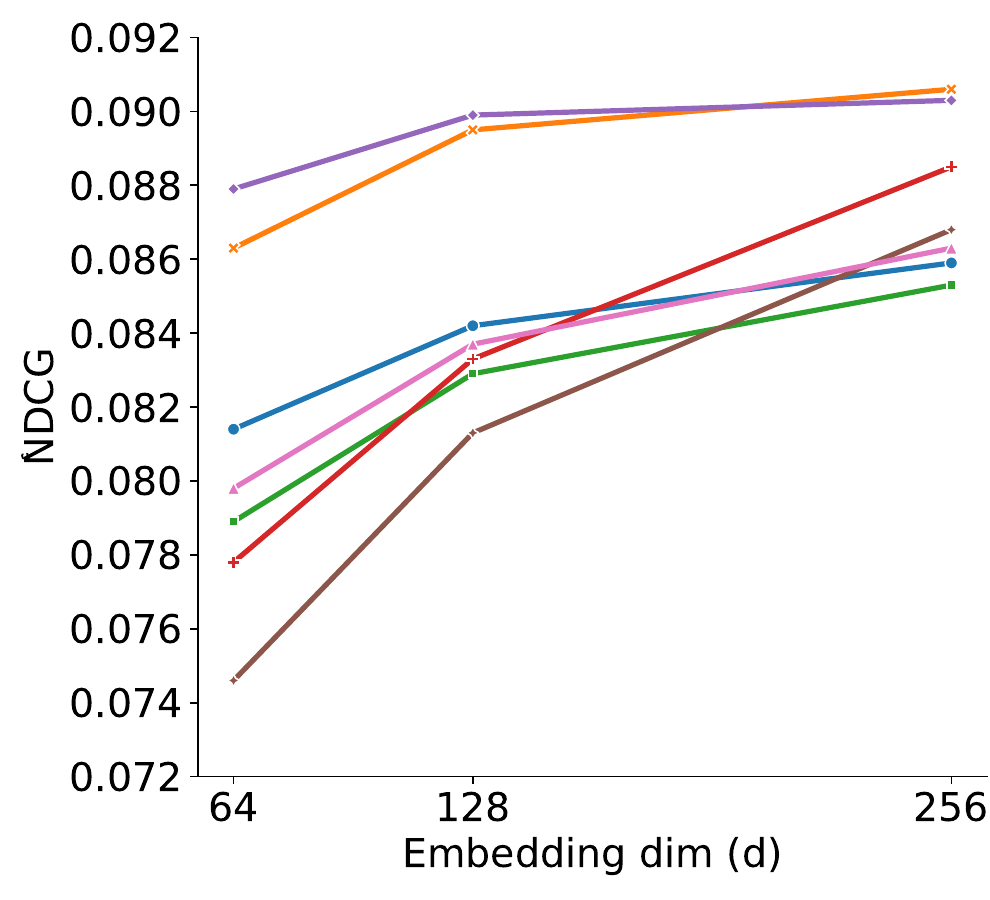}
		\caption{NDCG}
	\end{subfigure}
	\begin{subfigure}[b]{0.32\linewidth}
		\includegraphics[width=\textwidth]{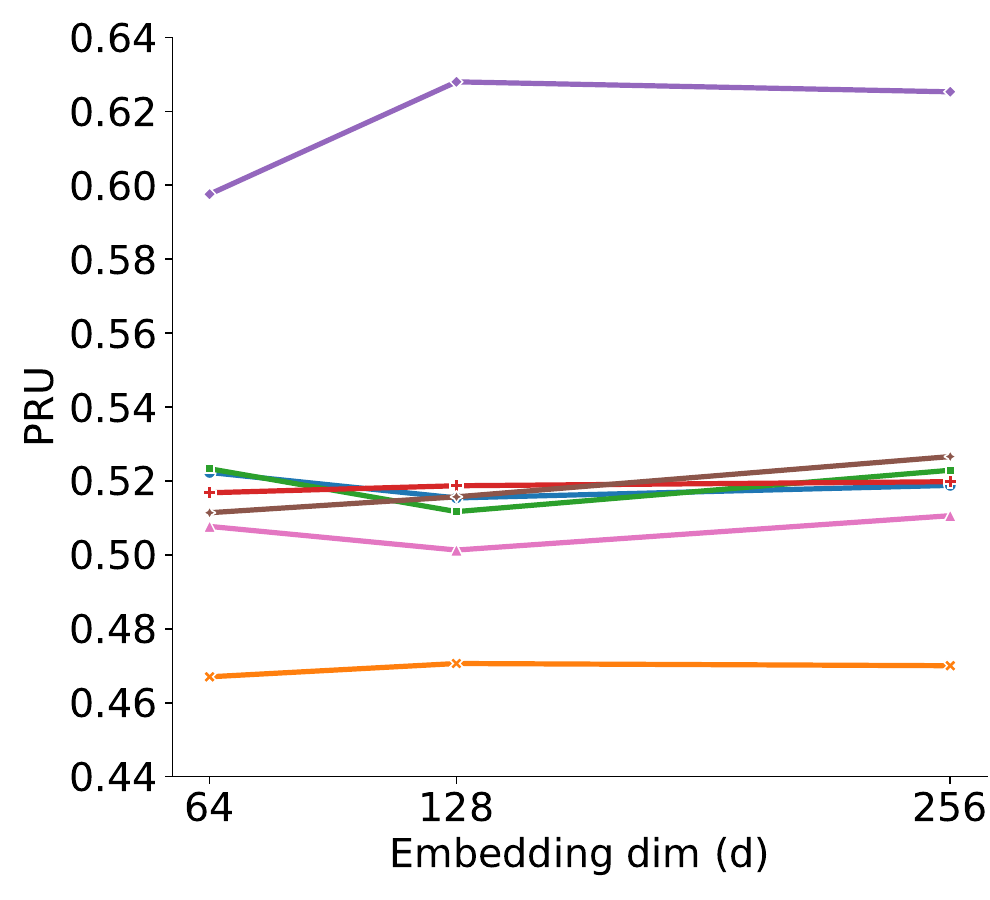}
		\caption{PRU}
	\end{subfigure}
	\begin{subfigure}[b]{0.32\linewidth}
		\includegraphics[width=\textwidth]{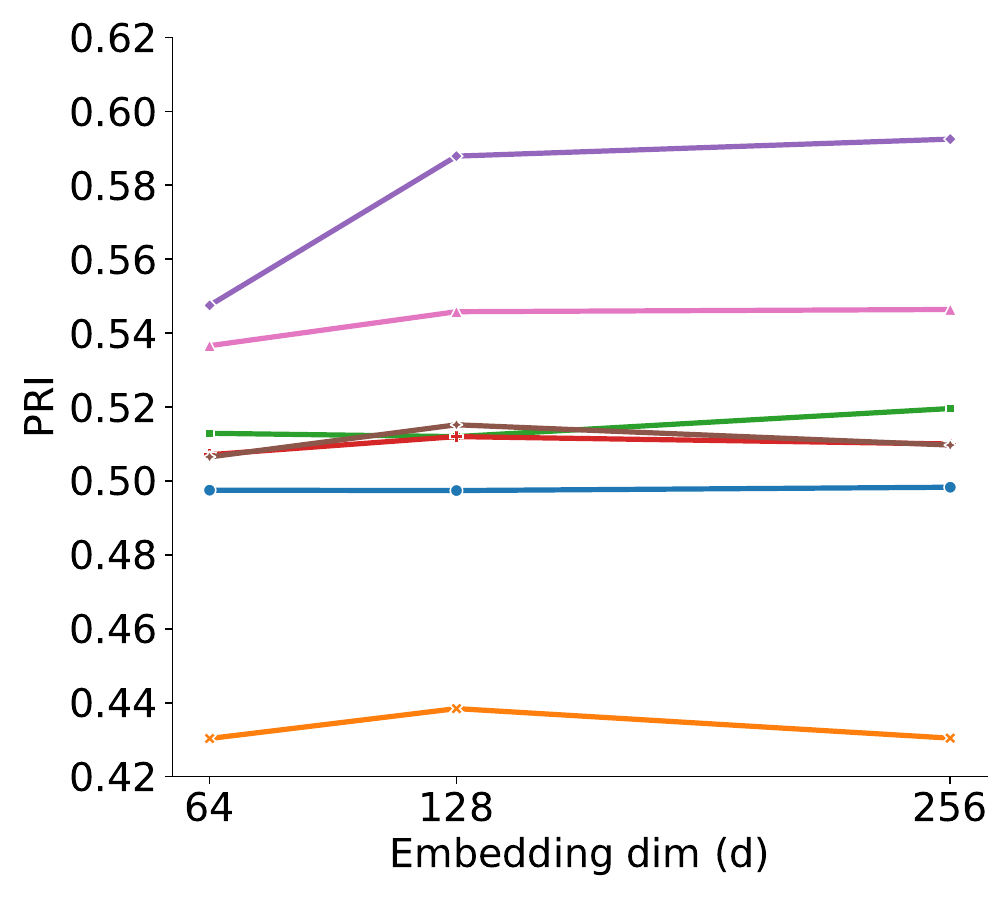}
		\caption{PRI}
	\end{subfigure}
	\caption{Comparing the models over the Epinions dataset, for different values of embeddings dimension.\label{fig:embed_dim_epinion}}
\end{figure}

\paragraph{Embeddings dimension}
The embeddings dimension is an important factor in  graph-based recommendation systems, representing the latent factors size of users and items.
We consider two datasets CDs and Epinions,
and compare the performance of the models across three embeddings sizes: 
$64$ (similar to LightGCN and r-AdjNorm),
$128$ (the default value for HetroFair)
and $256$ (the used value by APDA).
We utilize NDCG as the accuracy metric and PRU and PRI as the fairness metrics. 
The results are presented in Figures \ref{fig:embed_dim_cds} and \ref{fig:embed_dim_epinion}.
As can be seen in the figures,
over both datasets and by increasing the embeddings dimension,
NDCG of HetroFair increases, too.
In other words, in our model NDCG has positive correlation with the embeddings dimension.
For all values of the embeddings dimension,
HetroFair outperforms
the other methods,
in terms of both accuracy and fairness measures.
It can also be seen in the figures that the
dependence of PRI on the dimension size is negligible.
Considering the NDCG metric in Figures~\ref{fig:embed_dim_cds} and~\ref{fig:embed_dim_epinion},
the plots exhibit steeper slopes when the dimensions change from $64$ to $128$,
compared to the cases where the dimensions change from $128$ to $264$.
This means that the rates of accuracy changes decrease from 128 dimensions and beyond. According to our time complexity analysis in Section~\ref{sec:ourmethod},
time complexity of HetroFair 
increases linearly with the dimension of the embeddings.
Therefore, considering $128$ as the default dimension for embeddings seems to be a good trade-off between accuracy and complexity.

In the investigation of the baseline models,
it can be pointed out that no model consistently outperforms the other baselines
and by revising the parameters,
the relative performances of the models change with respect to each other.
For example in Figure \ref{fig:embed_dim_cds},
we see that in most cases, LightGCN achieves better fairness scores than APDA,
suggesting that the simple LightGCN method effectively captures data characteristics, particularly when a larger embeddings dimension is used.
On the other hand, r-AdjNorm has the worst performance in fairness metrics,
indicating that it fails to make a trade-off between accuracy and fairness.

\paragraph{\(Parameter \delta\)}
Parameter \(\delta\) plays a crucial role in the accuracy-fairness trade-off.
We consider the datasets Health and Beauty,
and change the value of \(\delta\)
from $0.05$ to $0.95$ with $0.05$ step.
The results are depicted in Figures \ref{fig:delta2} and \ref{fig:delta1}.
Over both datasets, we observe almost the same pattern:
within first few steps,
the fairness metrics are in their lowest values,
while the accuracy metrics are also too low.
As the value of \(\delta\) increases,
the performance of HetroFair increases,
while PRU and RRI deteriorate until reaching to an optimal value for \(\delta\),
which leads to almost stable values for all the metrics. 

\begin{figure}[H]
	\begin{subfigure}[b]{0.49\linewidth}
		\includegraphics[width=\linewidth]{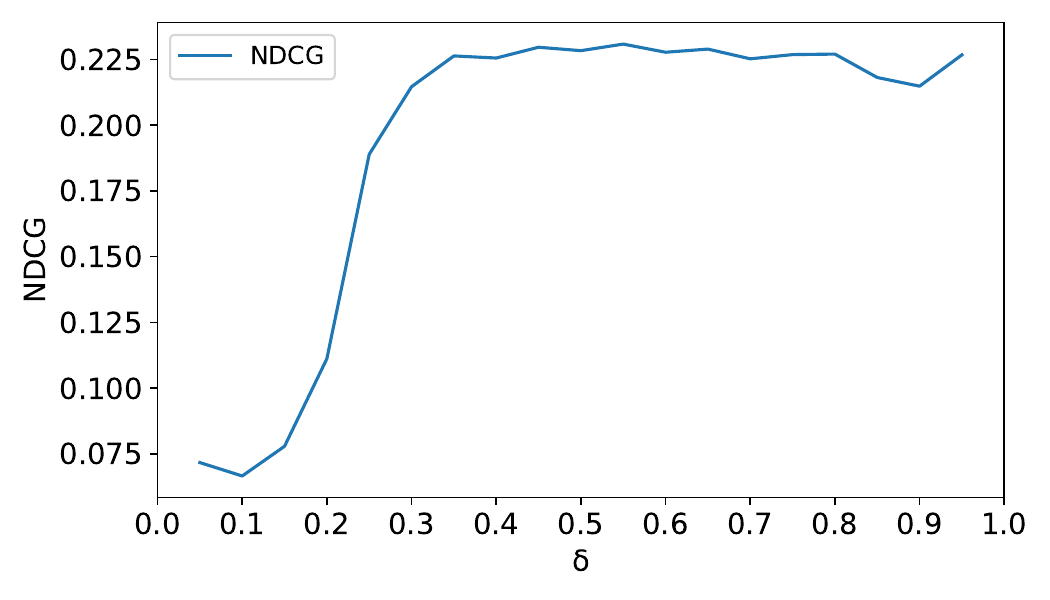}
		\caption{}
	\end{subfigure}
	\begin{subfigure}[b]{0.49\linewidth}
		\includegraphics[width=\linewidth]{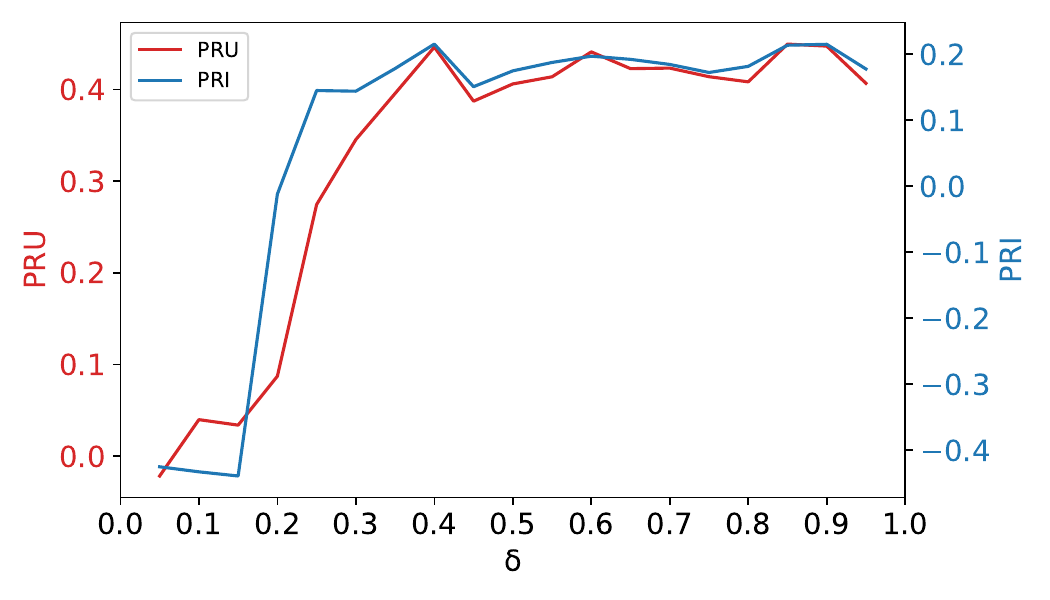}
		\caption{}
	\end{subfigure}
	\caption{Studying the effect of $\delta$, over the Beauty dataset.\label{fig:delta2}}
\end{figure}

\begin{figure}[H]
	\begin{subfigure}[b]{0.49\linewidth}
		\includegraphics[width=\linewidth]{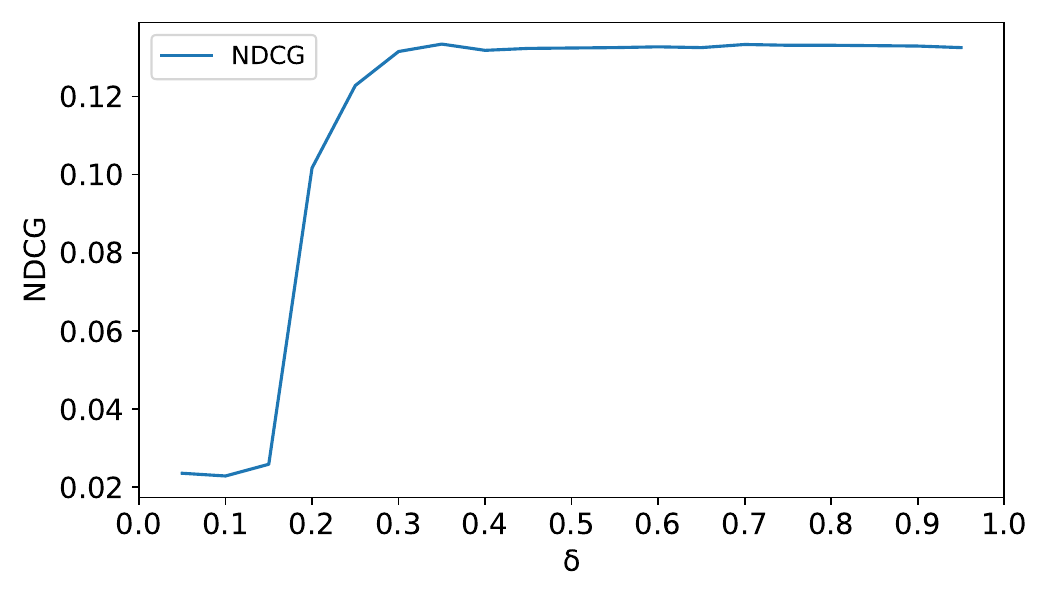}
		\caption{}
	\end{subfigure}
	\begin{subfigure}[b]{0.49\linewidth}
		\includegraphics[width=\linewidth]{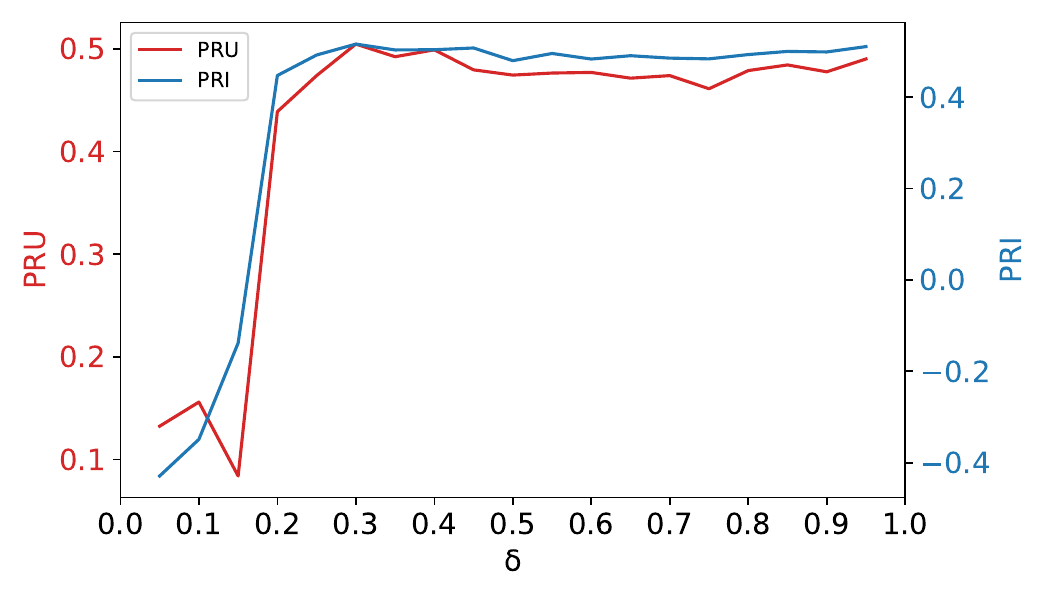}
		\caption{}
	\end{subfigure}
	\caption{Studying the effect of $\delta$, over the Health dataset.\label{fig:delta1}}
\end{figure}

\begin{figure*}[!t]
	\centering
	\begin{subfigure}[t]{0.2\linewidth}
		\includegraphics[width=\linewidth]{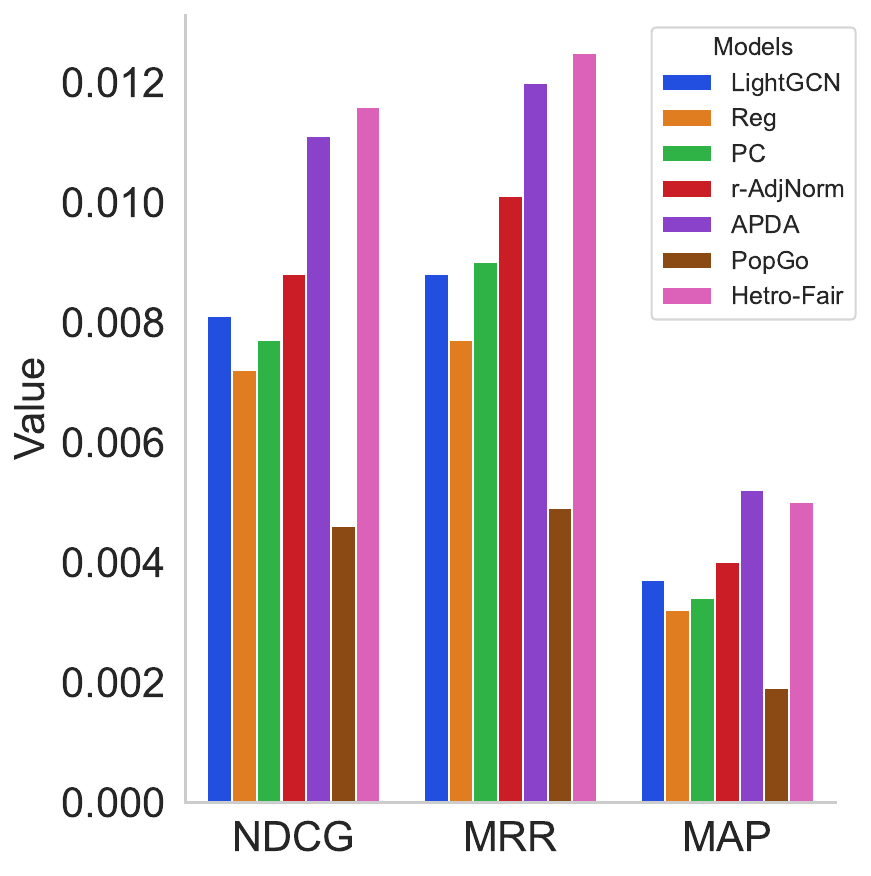}
		\caption{Long-tail items performance on the Epinions dataset.}
	\end{subfigure}
	\begin{subfigure}[t]{0.2\linewidth}
		\includegraphics[width=\linewidth]{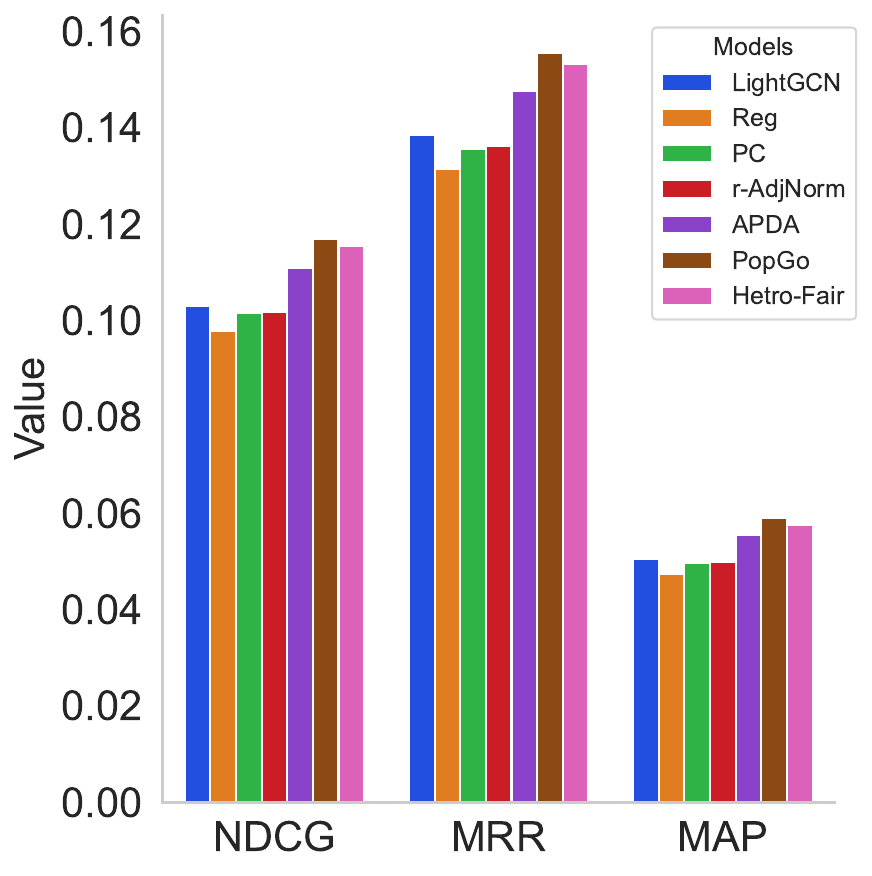}
		\caption{Short-head items performance on the Epinions dataset.}
	\end{subfigure}
	\begin{subfigure}[t]{0.2\linewidth}
		\includegraphics[width=\linewidth]{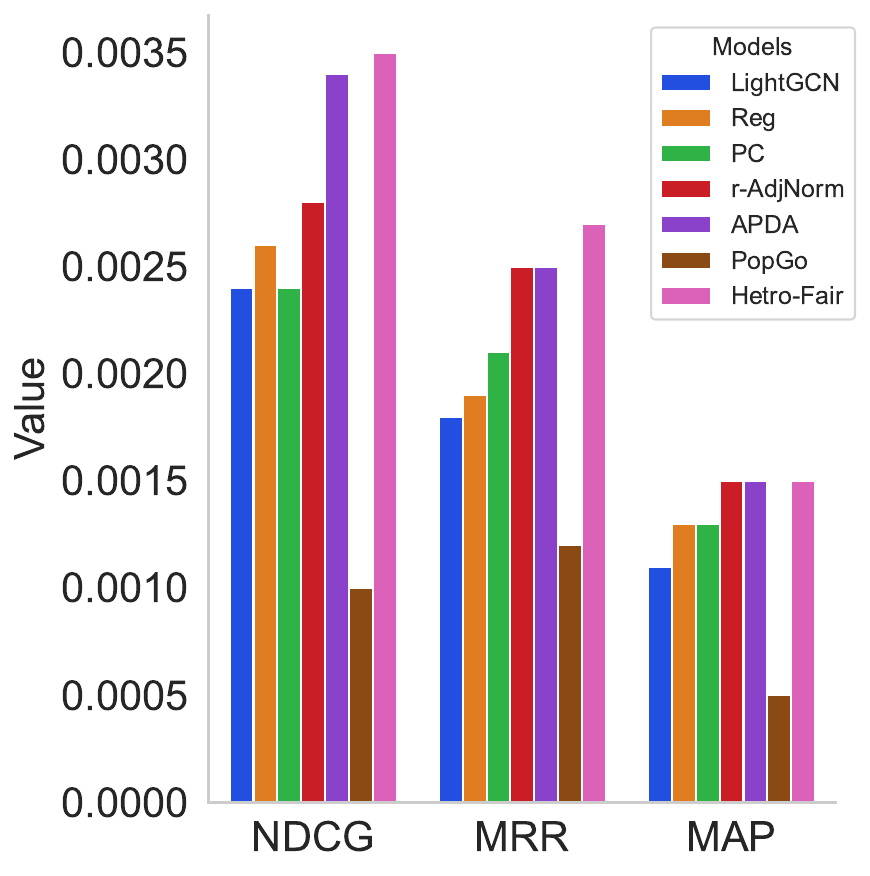}
		\caption{Long-tail items performance on the Electronics dataset.}
	\end{subfigure}
	\begin{subfigure}[t]{0.2\linewidth}
		\includegraphics[width=\linewidth]{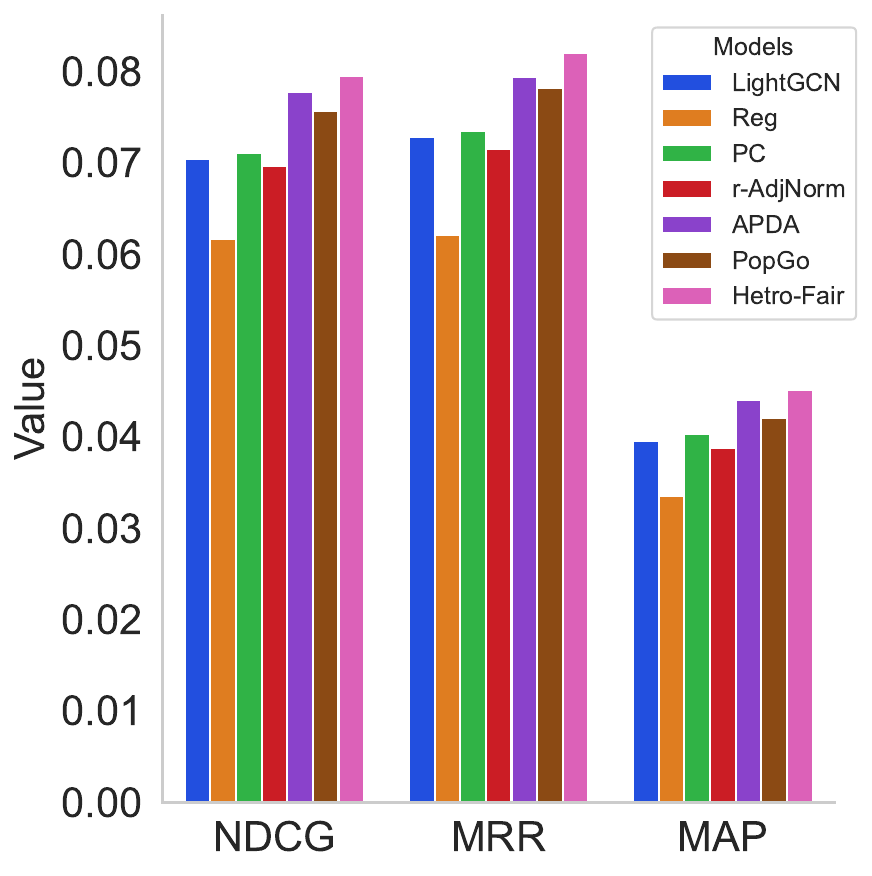}
		\caption{Short-head items performance on the Electronics dataset.}
	\end{subfigure}

	\begin{subfigure}[t]{0.2\linewidth}
		\includegraphics[width=\linewidth]{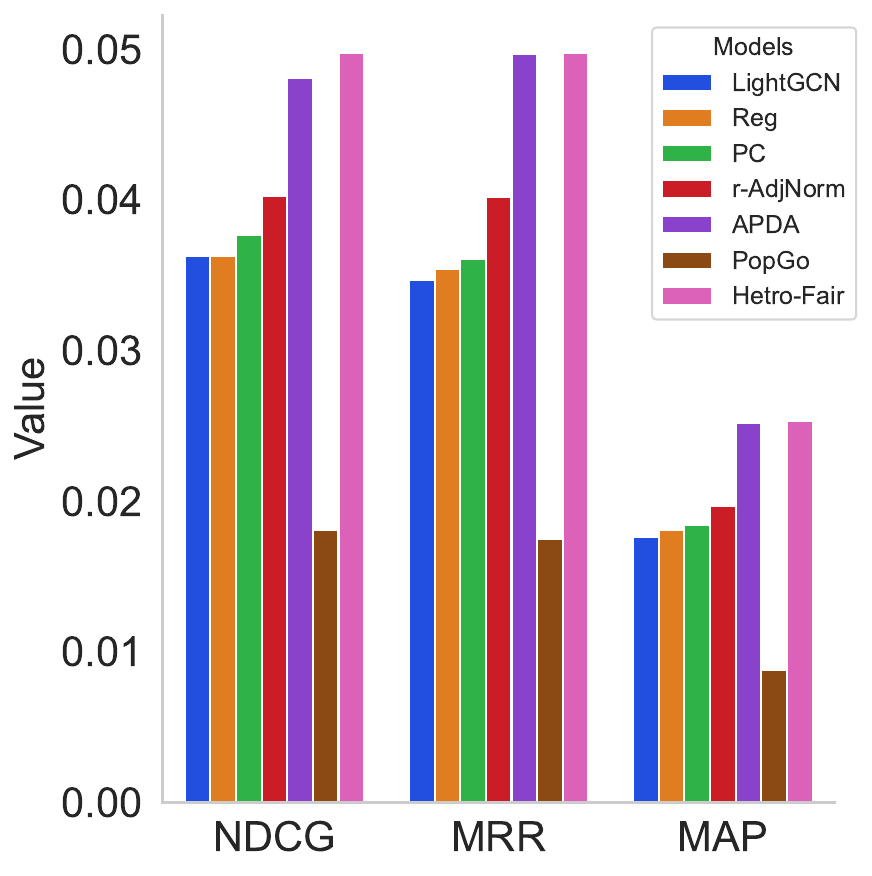}
		\caption{Long-tail items performance on the CDs dataset.}
	\end{subfigure}
	\begin{subfigure}[t]{0.2\linewidth}
		\includegraphics[width=\linewidth]{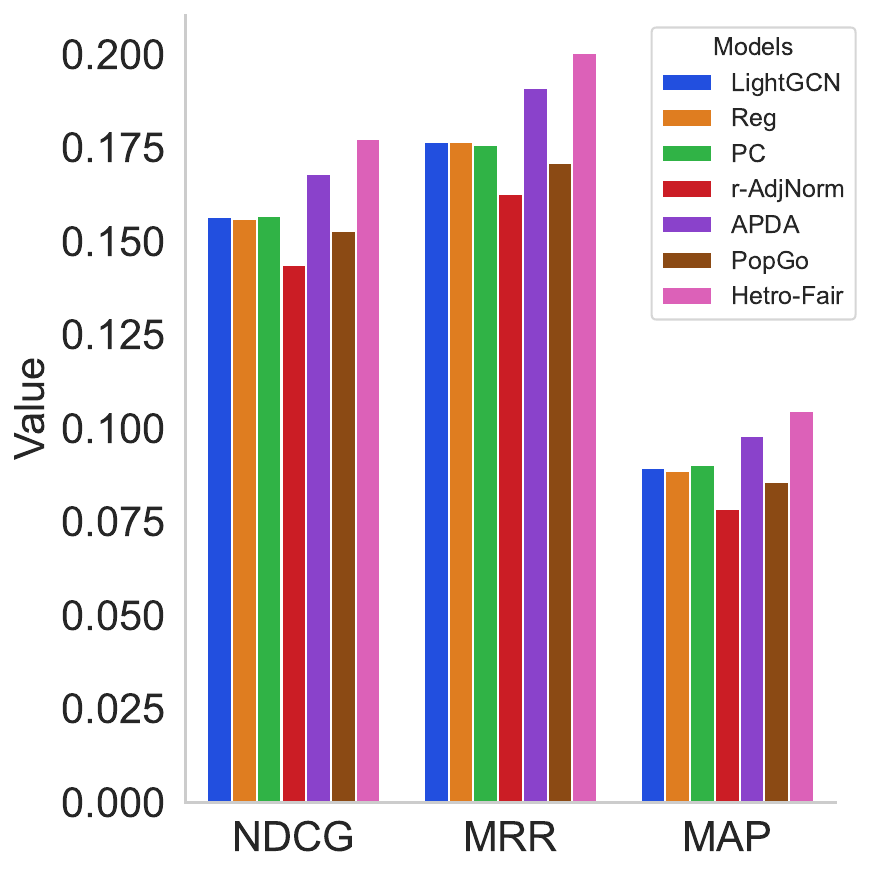}
		\caption{Short-head items performance on the CDs dataset.}
	\end{subfigure}
	\centering
	\begin{subfigure}[t]{0.2\linewidth}
		\includegraphics[width=\linewidth]{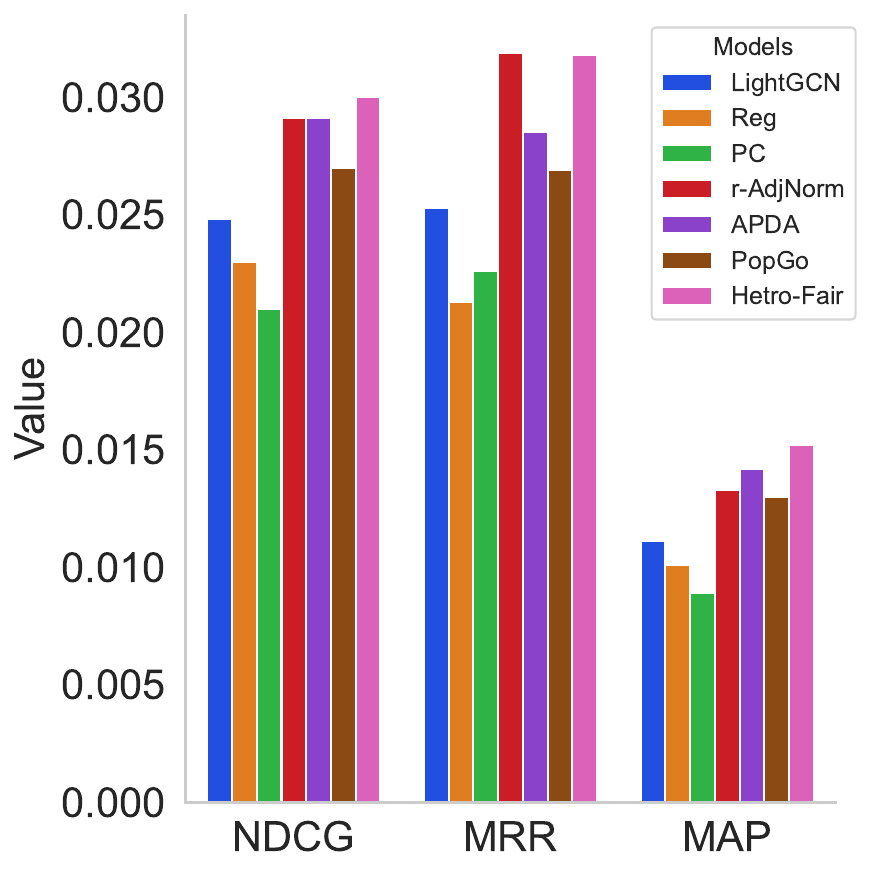}
		\caption{Long-tail items performance on the Health dataset.}
	\end{subfigure}
	\begin{subfigure}[t]{0.2\linewidth}
		\includegraphics[width=\linewidth]{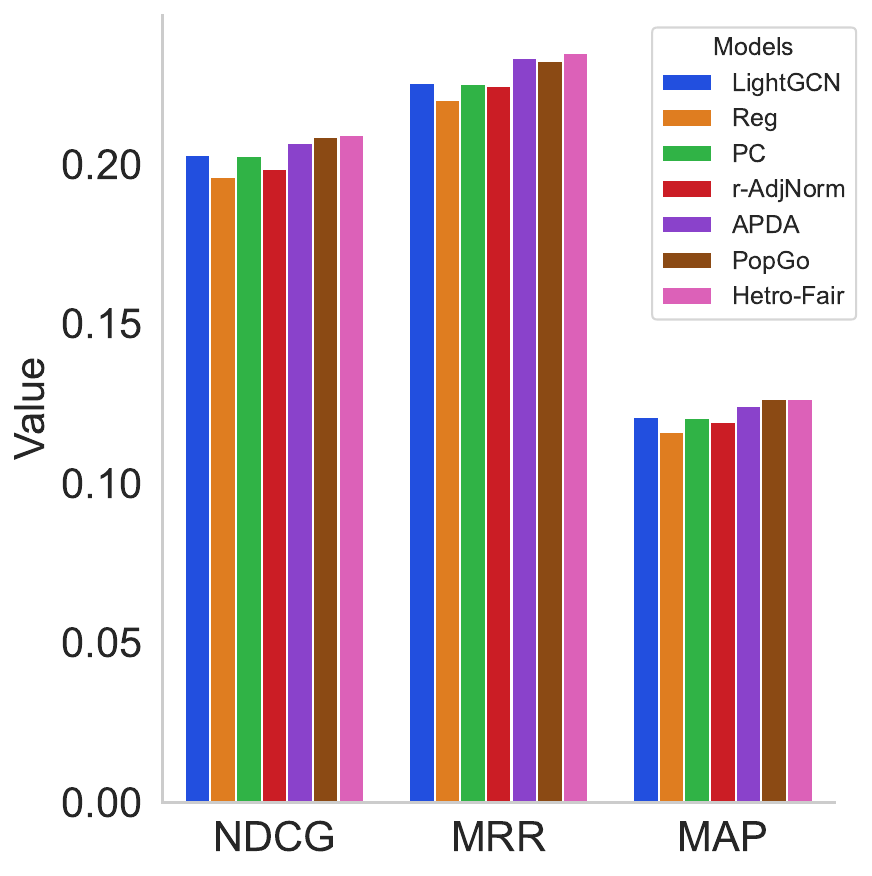}
		\caption{Short-head items performance on the Health dataset.}
	\end{subfigure}

	\begin{subfigure}[t]{0.2\linewidth}
		\includegraphics[width=\linewidth]{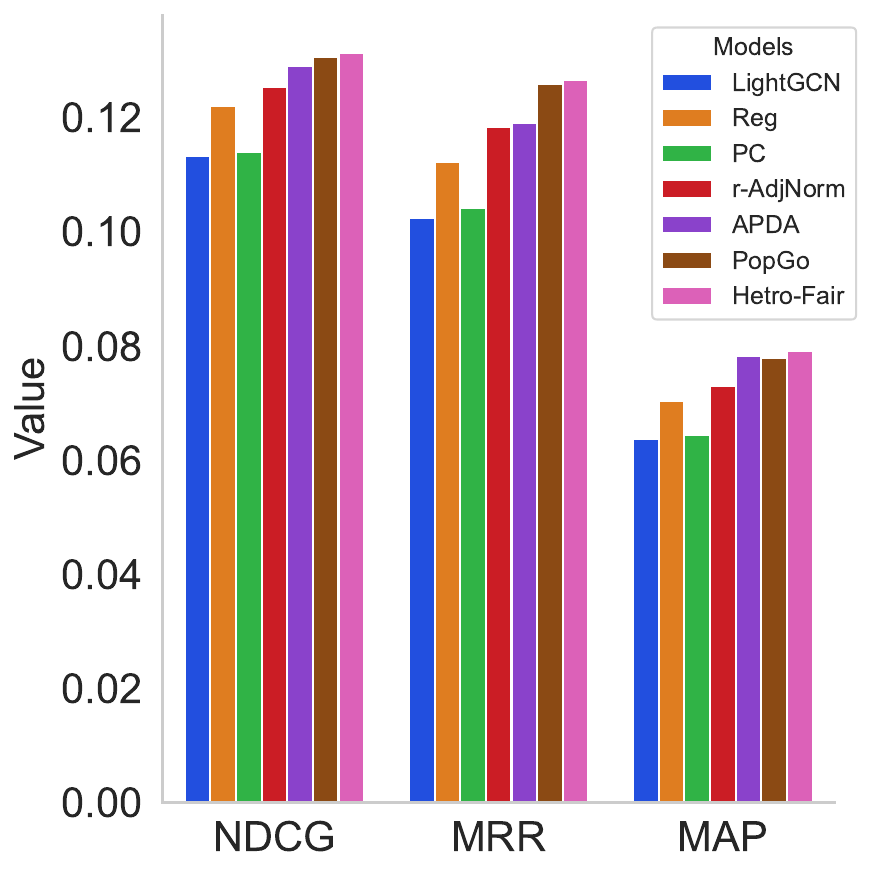}
		\caption{Long-tail items performance on the Beauty dataset.}
	\end{subfigure}
	\begin{subfigure}[t]{0.2\linewidth}
		\includegraphics[width=\linewidth]{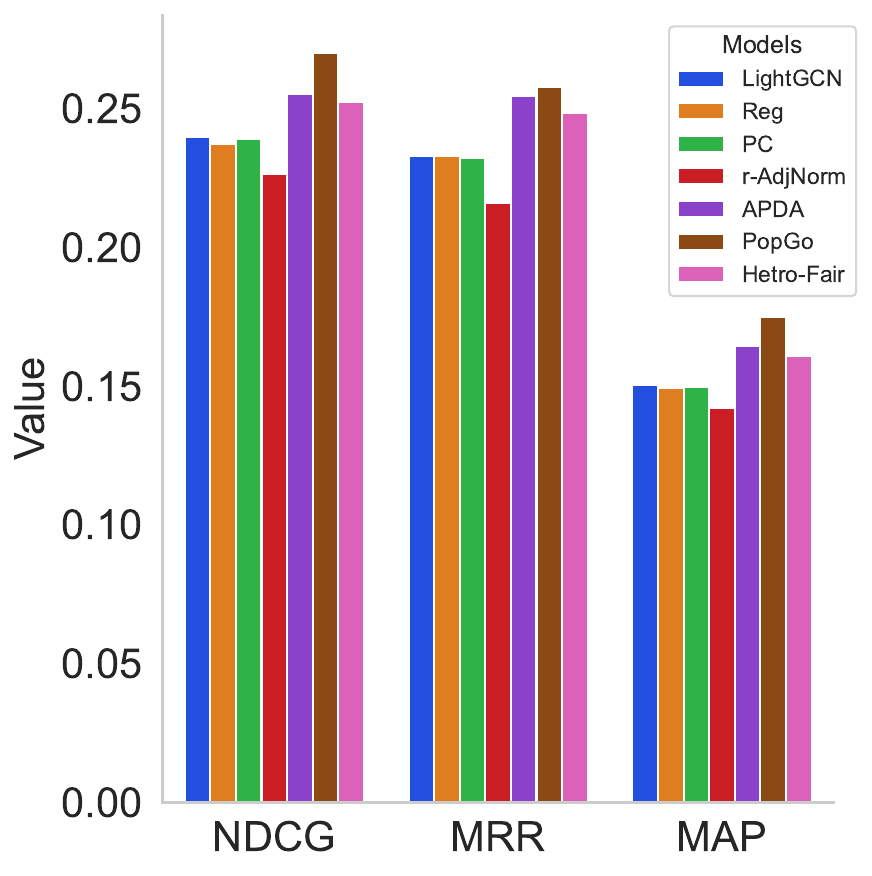}
		\caption{Short-head items performance on the Beauty dataset.}
	\end{subfigure}
	\begin{subfigure}[t]{0.2\linewidth}
		\includegraphics[width=\linewidth]{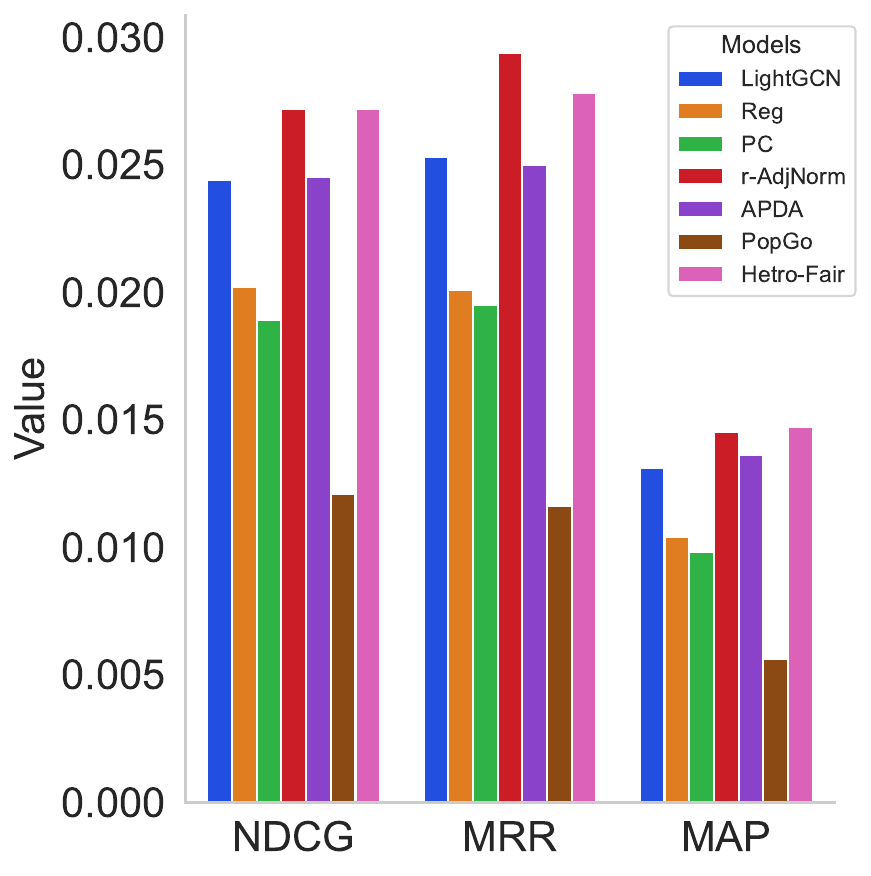}
		\caption{Long-tail items performance on the Movies dataset.}
	\end{subfigure}
	\begin{subfigure}[t]{0.2\linewidth}
		\includegraphics[width=\linewidth]{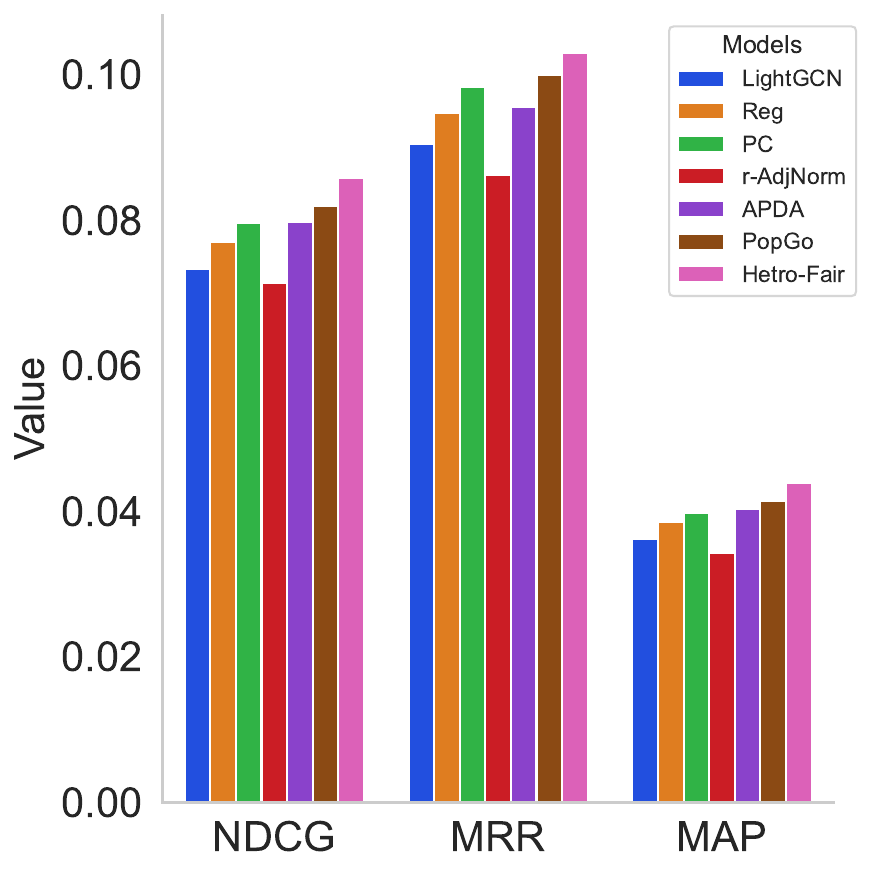}
		\caption{Short-head items performance on the Movies dataset.}
	\end{subfigure}
	\caption{Performance comparison of the methods,
		for short-head and long-tail items.\label{fig:long_short_performance}}
\end{figure*}

\subsection{Short-head and long-tail items performance}

To evaluate the effectiveness of HetroFair in recommending long-tail items,
we modify the prior test set and construct two following test sets,
one for long-tail items and the other for short-head items.
For the long-tail test set,
we exclude all short-head items from
the \(\tilde{O}_u^+\) sets of all users $u$.
For the short-head test set,
we remove all long-tail items existing in the \(\tilde{O}_u^+\) sets.
We sort the items based on their degree in descending order and consider the top \(20\%\) of them as short-head items, treating the rest as long-tail items.

Figure \ref{fig:long_short_performance} depicts the results of this investigation,
over all the six datasets.
These figures illustrate that for the category of long-tail items,
in most cases HetroFair achieves better accuracies.
This is compatible with the results presented in
Table~\ref{tbl:performance_comparison}, 
which indicate a reduction in the fairness metrics.
Moreover, this improvement coincides with an increase in accuracy in the category of short-head items.
Combining the results of Table~\ref{tbl:performance_comparison}
with the results presented in Figure \ref{fig:long_short_performance}
depicts that our performance improvement is not
solely tied to
long-tail items.
The simultaneous improvement of HetroFair in recommending short-head items
also plays a substantial role in reinforcing its overall performance.

\section{Conclusion}
\label{sec:conclusion}

In this paper, we studied item-side fairness in GNN-based recommendation systems. First, we pointed out that message normalization and aggregation processes in GNNs could lead to popularity bias. Then, we proposed the HetroFair model, which utilizes two techniques to overcome this bias and generate fair embeddings. The first technique is fairness-aware attention, which incorporates the dot product in the normalization process of GNNs to decrease the effect of nodes' degrees. The second technique is a trainable dot product similarity function, which assigns distinct weights to different features during the aggregation process. Our extensive experiments over six well-known datasets demonstrated the high performance of HetroFair in improving both user-side accuracy and item-side fairness.

\section*{Acknowledgment}

This work is supported by the Iran National Science Foundation (INSF)
under project No.4034377.

\bibliography{cite}

\begin{thebibliography}{10}

\bibitem{survey_multistakeholder}
{\sc Abdollahpouri, H., Adomavicius, G., Burke, R., Guy, I., Jannach, D.,
  Kamishima, T., Krasnodebski, J., and Pizzato, L.}
\newblock Multistakeholder recommendation: Survey and research directions.
\newblock {\em User Modeling and User-Adapted Interaction 30\/} (2020),
  127--158.

\bibitem{multistakeholder}
{\sc Abdollahpouri, H., and Burke, R.}
\newblock Multistakeholder recommender systems.
\newblock In {\em Recommender systems handbook}. Springer, 2021, pp.~647--677.

\bibitem{abdollahpouri2019unfairness}
{\sc Abdollahpouri, H., Mansoury, M., Burke, R., and Mobasher, B.}
\newblock The unfairness of popularity bias in recommendation.
\newblock {\em arXiv:1907.13286\/} (2019).

\bibitem{cp_ecir}
{\sc Anelli, V.~W., Deldjoo, Y., Di~Noia, T., Malitesta, D., Paparella, V., and
  Pomo, C.}
\newblock Auditing consumer-and producer-fairness in graph collaborative
  filtering.
\newblock In {\em European Conference on Information Retrieval\/} (2023),
  Springer, pp.~33--48.

\bibitem{gcmc}
{\sc Berg, R. v.~d., Kipf, T.~N., and Welling, M.}
\newblock Graph convolutional matrix completion.
\newblock {\em arXiv preprint arXiv:1706.02263\/} (2017).

\bibitem{halfdecade}
{\sc Chehreghani, M.~H.}
\newblock Half a decade of graph convolutional networks.
\newblock {\em Nature Machine Intelligence 4}, 3 (2022), 192--193.

\bibitem{chen2024graph}
{\sc Chen, J., Wu, J., Chen, J., Xin, X., Li, Y., and He, X.}
\newblock How graph convolutions amplify popularity bias for recommendation?
\newblock {\em Frontiers of Computer Science 18}, 5 (2024), 185603.

\bibitem{lrgccf}
{\sc Chen, L., Wu, L., Hong, R., Zhang, K., and Wang, M.}
\newblock Revisiting graph based collaborative filtering: A linear residual
  graph convolutional network approach.
\newblock In {\em AAAI conference on artificial intelligence\/} (2020),
  vol.~34, pp.~27--34.

\bibitem{ecir2023item}
{\sc D’Amico, E., Muhammad, K., Tragos, E., Smyth, B., Hurley, N., and
  Lawlor, A.}
\newblock Item graph convolution collaborative filtering for inductive
  recommendations.
\newblock In {\em European Conference on Information Retrieval\/} (2023),
  Springer, pp.~249--263.

\bibitem{gnnsurvey}
{\sc Gao, C., Zheng, Y., Li, N., Li, Y., Qin, Y., Piao, J., Quan, Y., Chang,
  J., Jin, D., He, X., et~al.}
\newblock A survey of graph neural networks for recommender systems:
  Challenges, methods, and directions.
\newblock {\em ACM Transactions on Recommender Systems 1}, 1 (2023), 1--51.

\bibitem{xavier}
{\sc Glorot, X., and Bengio, Y.}
\newblock Understanding the difficulty of training deep feedforward neural
  networks.
\newblock In {\em Proceedings of the thirteenth international conference on
  artificial intelligence and statistics\/} (2010), pp.~249--256.

\bibitem{graphsage}
{\sc Hamilton, W., Ying, Z., and Leskovec, J.}
\newblock Inductive representation learning on large graphs.
\newblock {\em Advances in neural information processing systems 30\/} (2017).

\bibitem{amazon-data}
{\sc He, R., and McAuley, J.}
\newblock Ups and downs: Modeling the visual evolution of fashion trends with
  one-class collaborative filtering.
\newblock In {\em proceedings of the 25th international conference on world
  wide web\/} (2016), pp.~507--517.

\bibitem{lightgcn}
{\sc He, X., Deng, K., Wang, X., Li, Y., Zhang, Y., and Wang, M.}
\newblock Lightgcn: Simplifying and powering graph convolution network for
  recommendation.
\newblock In {\em Proceedings of the 43rd International ACM SIGIR conference on
  research and development in Information Retrieval\/} (2020), pp.~639--648.

\bibitem{isufiipm}
{\sc Isufi, E., Pocchiari, M., and Hanjalic, A.}
\newblock Accuracy-diversity trade-off in recommender systems via graph
  convolutions.
\newblock {\em Information Processing \& Management 58}, 2 (2021), 102459.

\bibitem{gcn}
{\sc Kipf, T.~N., and Welling, M.}
\newblock Semi-supervised classification with graph convolutional networks.
\newblock {\em arXiv:1609.02907\/} (2016).

\bibitem{hashtag}
{\sc Li, M., Gan, T., Liu, M., Cheng, Z., Yin, J., and Nie, L.}
\newblock Long-tail hashtag recommendation for micro-videos with graph
  convolutional network.
\newblock In {\em Proceedings of the 28th ACM International Conference on
  Information and Knowledge Management\/} (2019), pp.~509--518.

\bibitem{ultragcn}
{\sc Mao, K., Zhu, J., Xiao, X., Lu, B., Wang, Z., and He, X.}
\newblock Ultragcn: ultra simplification of graph convolutional networks for
  recommendation.
\newblock In {\em 30th ACM International Conference on Information \& Knowledge
  Management\/} (2021), pp.~1253--1262.

\bibitem{DBLP:journals/corr/abs-2401-01384}
{\sc Mohamadi, Y., and Chehreghani, M.~H.}
\newblock Strong transitivity relations and graph neural networks.
\newblock {\em CoRR abs/2401.01384\/} (2024).

\bibitem{naghiaei2022unfairness}
{\sc Naghiaei, M., Rahmani, H.~A., and Dehghan, M.}
\newblock The unfairness of popularity bias in book recommendation.
\newblock In {\em International Workshop on Algorithmic Bias in Search and
  Recommendation\/} (2022), Springer, pp.~69--81.

\bibitem{cpfair}
{\sc Naghiaei, M., Rahmani, H.~A., and Deldjoo, Y.}
\newblock Cpfair: Personalized consumer and producer fairness re-ranking for
  recommender systems.
\newblock In {\em Proceedings of the 45th International ACM SIGIR Conference on
  Research and Development in Information Retrieval\/} (2022), pp.~770--779.

\bibitem{10.1145/3700790}
{\sc Nasrabadi, F.~G., Kashani, A., Zahedi, P., and Chehreghani, M.~H.}
\newblock Content augmented graph neural networks.
\newblock {\em ACM Trans. Web\/} (Oct. 2024).
\newblock Just Accepted.

\bibitem{poi_unfairness}
{\sc Rahmani, H.~A., Deldjoo, Y., Tourani, A., and Naghiaei, M.}
\newblock The unfairness of active users and popularity bias in
  point-of-interest recommendation.
\newblock In {\em International Workshop on Algorithmic Bias in Search and
  Recommendation\/} (2022), Springer, pp.~56--68.

\bibitem{bpr}
{\sc Rendle, S., Freudenthaler, C., Gantner, Z., and Schmidt-Thieme, L.}
\newblock Bpr: Bayesian personalized ranking from implicit feedback.
\newblock {\em arXiv preprint arXiv:1205.2618\/} (2012).

\bibitem{powerfulgraph}
{\sc Shen, Y., Wu, Y., Zhang, Y., Shan, C., Zhang, J., Letaief, B.~K., and Li,
  D.}
\newblock How powerful is graph convolution for recommendation?
\newblock In {\em Proceedings of the 30th ACM international conference on
  information \& knowledge management\/} (2021), pp.~1619--1629.

\bibitem{bgcf}
{\sc Sun, J., Guo, W., Zhang, D., Zhang, Y., Regol, F., Hu, Y., Guo, H., Tang,
  R., Yuan, H., He, X., et~al.}
\newblock A framework for recommending accurate and diverse items using
  bayesian graph convolutional neural networks.
\newblock In {\em Proceedings of the 26th ACM SIGKDD International Conference
  on Knowledge Discovery \& Data Mining\/} (2020), pp.~2030--2039.

\bibitem{gat}
{\sc Veli{\v{c}}kovi{\'c}, P., Cucurull, G., Casanova, A., Romero, A., Lio, P.,
  and Bengio, Y.}
\newblock Graph attention networks.
\newblock {\em arXiv preprint arXiv:1710.10903\/} (2017).

\bibitem{m2grl}
{\sc Wang, M., Lin, Y., Lin, G., Yang, K., and Wu, X.-m.}
\newblock M2grl: A multi-task multi-view graph representation learning
  framework for web-scale recommender systems.
\newblock In {\em Proceedings of the 26th ACM SIGKDD international conference
  on knowledge discovery \& data mining\/} (2020), pp.~2349--2358.

\bibitem{ngcf}
{\sc Wang, X., He, X., Wang, M., Feng, F., and Chua, T.-S.}
\newblock Neural graph collaborative filtering.
\newblock In {\em Proceedings of the 42nd international ACM SIGIR conference on
  Research and development in Information Retrieval\/} (2019), pp.~165--174.

\bibitem{dgcf}
{\sc Wang, X., Jin, H., Zhang, A., He, X., Xu, T., and Chua, T.-S.}
\newblock Disentangled graph collaborative filtering.
\newblock In {\em Proceedings of the 43rd international ACM SIGIR conference on
  research and development in information retrieval\/} (2020), pp.~1001--1010.

\bibitem{wei2021model}
{\sc Wei, T., Feng, F., Chen, J., Wu, Z., Yi, J., and He, X.}
\newblock Model-agnostic counterfactual reasoning for eliminating popularity
  bias in recommender system.
\newblock In {\em Proceedings of the 27th ACM SIGKDD conference on knowledge
  discovery \& data mining\/} (2021), pp.~1791--1800.

\bibitem{sgl}
{\sc Wu, J., Wang, X., Feng, F., He, X., Chen, L., Lian, J., and Xie, X.}
\newblock Self-supervised graph learning for recommendation.
\newblock In {\em 44th international ACM SIGIR conference on research and
  development in information retrieval\/} (2021), pp.~726--735.

\bibitem{xie2021improving}
{\sc Xie, R., Liu, Q., Liu, S., Zhang, Z., Cui, P., Zhang, B., and Lin, L.}
\newblock Improving accuracy and diversity in matching of recommendation with
  diversified preference network.
\newblock {\em IEEE Transactions on Big Data 8}, 4 (2021), 955--967.

\bibitem{DBLP:conf/iclr/XuHLJ19}
{\sc Xu, K., Hu, W., Leskovec, J., and Jegelka, S.}
\newblock How powerful are graph neural networks?
\newblock In {\em 7th International Conference on Learning Representations,
  {ICLR} 2019, New Orleans, LA, USA, May 6-9, 2019\/} (2019), OpenReview.net.

\bibitem{dgrec}
{\sc Yang, L., Wang, S., Tao, Y., Sun, J., Liu, X., Yu, P.~S., and Wang, T.}
\newblock Dgrec: Graph neural network for recommendation with diversified
  embedding generation.
\newblock In {\em Proceedings of the Sixteenth ACM International Conference on
  Web Search and Data Mining\/} (2023), pp.~661--669.

\bibitem{pinsage}
{\sc Ying, R., He, R., Chen, K., Eksombatchai, P., Hamilton, W.~L., and
  Leskovec, J.}
\newblock Graph convolutional neural networks for web-scale recommender
  systems.
\newblock In {\em Proceedings of the 24th ACM SIGKDD international conference
  on knowledge discovery \& data mining\/} (2018), pp.~974--983.

\bibitem{graph-augment}
{\sc Yu, J., Yin, H., Xia, X., Chen, T., Cui, L., and Nguyen, Q. V.~H.}
\newblock Are graph augmentations necessary? simple graph contrastive learning
  for recommendation.
\newblock In {\em Proceedings of the 45th international ACM SIGIR conference on
  research and development in information retrieval\/} (2022), pp.~1294--1303.

\bibitem{yu2022low}
{\sc Yu, W., Zhang, Z., and Qin, Z.}
\newblock Low-pass graph convolutional network for recommendation.
\newblock In {\em AAAI Conference on Artificial Intelligence\/} (2022),
  vol.~36, pp.~8954--8961.

\bibitem{zhang2024robust}
{\sc Zhang, A., Ma, W., Zheng, J., Wang, X., and Chua, T.-S.}
\newblock Robust collaborative filtering to popularity distribution shift.
\newblock {\em ACM Transactions on Information Systems 42}, 3 (2024), 1--25.

\bibitem{bilateral}
{\sc Zhao, M., Deng, Q., Wang, K., Wu, R., Tao, J., Fan, C., Chen, L., and Cui,
  P.}
\newblock Bilateral filtering graph convolutional network for multi-relational
  social recommendation in the power-law networks.
\newblock {\em ACM Transactions on Information Systems 40}, 2 (2021), 1--24.

\bibitem{adjnorm}
{\sc Zhao, M., Wu, L., Liang, Y., Chen, L., Zhang, J., Deng, Q., Wang, K.,
  Shen, X., Lv, T., and Wu, R.}
\newblock Investigating accuracy-novelty performance for graph-based
  collaborative filtering.
\newblock In {\em 45th International ACM SIGIR Conference on Research and
  Development in Information Retrieval\/} (2022), pp.~50--59.

\bibitem{zhao2022popularity}
{\sc Zhao, Z., Chen, J., Zhou, S., He, X., Cao, X., Zhang, F., and Wu, W.}
\newblock Popularity bias is not always evil: Disentangling benign and harmful
  bias for recommendation.
\newblock {\em IEEE Transactions on Knowledge and Data Engineering\/} (2022).

\bibitem{dgcn}
{\sc Zheng, Y., Gao, C., Chen, L., Jin, D., and Li, Y.}
\newblock Dgcn: Diversified recommendation with graph convolutional networks.
\newblock In {\em Proceedings of the Web Conference 2021\/} (2021),
  pp.~401--412.

\bibitem{apda}
{\sc Zhou, H., Chen, H., Dong, J., Zha, D., Zhou, C., and Huang, X.}
\newblock Adaptive popularity debiasing aggregator for graph collaborative
  filtering.
\newblock In {\em Proceedings of the 46th International ACM SIGIR Conference on
  Research and Development in Information Retrieval\/} (2023), pp.~7--17.

\bibitem{zhou2020graph}
{\sc Zhou, J., Cui, G., Hu, S., Zhang, Z., Yang, C., Liu, Z., Wang, L., Li, C.,
  and Sun, M.}
\newblock Graph neural networks: A review of methods and applications.
\newblock {\em AI open 1\/} (2020), 57--81.

\bibitem{popularitymetrics}
{\sc Zhu, Z., He, Y., Zhao, X., Zhang, Y., Wang, J., and Caverlee, J.}
\newblock Popularity-opportunity bias in collaborative filtering.
\newblock In {\em Proceedings of the 14th ACM International Conference on Web
  Search and Data Mining\/} (2021), pp.~85--93.

\bibitem{DBLP:journals/tjs/ZohrabiSC24}
{\sc Zohrabi, M., Saravani, S., and Chehreghani, M.~H.}
\newblock Centrality-based and similarity-based neighborhood extension in graph
  neural networks.
\newblock {\em J. Supercomput. 80}, 16 (2024), 24638--24663.

\end{thebibliography}
\bibliographystyle{acm}
	
\end{document}